\documentclass[letterpaper, 10pt, journal, twocolumn]{IEEEtran}
\usepackage{epsfig}
\usepackage{amsmath,graphicx,amsfonts,amssymb,epsfig,subfig,mathrsfs,mathtools}
\usepackage{ntheorem}
\usepackage{cuted}
\usepackage{arydshln}
\usepackage{ulem}
\usepackage{blkarray}

\usepackage{color}
\usepackage{multirow}
\usepackage{multicol}
\usepackage{lipsum}
\usepackage{rotating}
\usepackage{graphicx}\usepackage{psfrag}
\usepackage{algorithm}
\usepackage[utf8]{inputenc}
\usepackage{xspace}

\usepackage{booktabs}
\usepackage{algorithmicx}
\usepackage{algpseudocode}
\usepackage{cite}
\usepackage{tikz}
\usepackage{cancel}
\usepackage{textcomp}
 \definecolor{lin}{RGB}{240,0,0}
 \definecolor{paleblue}{RGB}{0,9,255}
\usepackage[colorlinks=true,linkcolor=lin,citecolor=paleblue,pdfa]{hyperref}
\usepackage{comment}

\newcommand{\maps}[3]{#1: #2 \mapsto #3}
\newcommand{\map}[3]{#1: #2 \rightarrow #3}
\newcommand{\setdef}[2]{\{#1 \; | \; #2\}}
\newcommand{\st}{\ensuremath{\operatorname{s.t.}}}
\newcommand{\real}{\ensuremath{\mathbb{R}}}
\newcommand{\probP}{\ensuremath{\mathbb{P}}}
\newcommand{\probQ}{\ensuremath{\mathbb{Q}}}

\newcommand{\unitcircle}{\ensuremath{\mathbb{S}^1}}
\newcommand{\realnonnegative}{\ensuremath{\mathbb{R}}_{\ge 0}}
\newcommand{\integernonnegative}{\ensuremath{\mathbb{Z}}_{\ge 0}}
\newcommand{\integerpositive}{\ensuremath{\mathbb{Z}}_{> 0}}
\newcommand{\until}[1]{\{1,\dots, #1\}}
\newcommand{\subscr}[2]{#1_{\textup{#2}}}
\newcommand{\supscr}[2]{#1^{\textup{#2}}}
\newcommand{\vect}[1]{\boldsymbol{#1}}
\newcommand{\vectorones}[1]{\vect{1}_{#1}}
\newcommand{\vectorzeros}[1]{\vect{0}_{#1}}
\newcommand{\Norm}[1]{\|#1\|}
\newcommand{\trans}[1]{{#1}^\top}

\newcommand{\Prob}{\ensuremath{\operatorname{Prob}}}

\newcommand{\dom}{\ensuremath{\operatorname{dom}}}

\newcommand{\untilinterval}[2]{\{{#1},\dots, {#2}\}}

\newcommand{\Lip}{\ensuremath{\operatorname{Lip}}}
\newcommand{\ow}{\ensuremath{\operatorname{otherwise}}}

\newcommand{\qed}{~\hfill{ $\blacksquare$}}

\graphicspath{{./Fig/}}
\usepackage{epstopdf}
\DeclareMathOperator*{\argmin}{argmin}

\makeatletter
\newtheoremstyle{breaknote}  {\item{\theorem@headerfont
          ##1\ ##2\theorem@separator}\hskip\labelsep\relax}  {\item{\theorem@headerfont
          ##1\ ##2\ (##3)\theorem@separator}\hskip\labelsep\relax}
\makeatother
\theoremstyle{breaknote}
\newtheorem{assumption}{Assumption}[section]
\newtheorem{theorem}{Theorem}[section]

\newtheorem{definition}{Definition}[section]
\newtheorem{lemma}{Lemma}[section]

\newtheorem{remark}{Remark}[section]

\renewcommand{\baselinestretch}{.97}

\title{Online Optimization and Ambiguity-based Learning of Distributionally Uncertain Dynamic Systems}

\author{Dan Li$^{1}$, Dariush Fooladivanda$^{1}$ and Sonia
Mart{\'\i}nez$^{1}$ \thanks{*
This research was developed with funding from ONR N00014-19-1-2471, and AFOSR FA9550-19-1-0235. }
\thanks{$^{1}$ D. Li and S. Mart{\'\i}nez are with the Department of
Mechanical and Aerospace Engineering, University of California San
Diego, La Jolla, CA 92092, USA. D. Fooladivanda is with the
Department of Electrical Engineering and Computer Sciences,
University of California at Berkeley, Berkeley, CA 94720, USA. {\tt
dal027.me@gmail.com; dfooladi@berkeley.edu; soniamd@ucsd.edu}}}

\begin{document}

\maketitle
\begin{abstract}
  This paper proposes a novel approach to construct data-driven online solutions to
  optimization
  problems (P)
  subject to a class of distributionally uncertain dynamical
  systems. The introduced framework allows for the simultaneous learning of distributional system uncertainty via a parameterized, control-dependent ambiguity set using a finite historical data set, and its use to make online decisions with probabilistic regret function bounds. Leveraging the merits of Machine Learning, the main technical approach relies on the theory of Distributional Robust Optimization (DRO), to hedge against uncertainty and provide less conservative results than standard Robust Optimization approaches.
  Starting from recent results that describe ambiguity sets via parameterized, and control-dependent empirical distributions as well as ambiguity radii, we first present a tractable reformulation of the corresponding optimization problem while maintaining the probabilistic guarantees. We then specialize these problems to the cases of 1) optimal one-stage control of distributionally uncertain nonlinear systems, and 2) resource allocation under distributional uncertainty. A novelty of this work is that it
    extends DRO to online optimization problems subject to a  distributionally uncertain dynamical system constraint, handled via a control-dependent ambiguity set that leads to
    online-tractable optimization with probabilistic guarantees on regret bounds.
  Further, we introduce an online version of the
  Nesterov's accelerated-gradient algorithm, and analyze its performance  to solve this class of problems via dissipativity theory.          \end{abstract}

\section{Introduction}
Online optimization has attracted significant attention from various
fields, including Machine Learning,
Information Theory, Robotics  and Smart Power Systems; see~\cite{SSS:12,AR-KS-AT:10,AS-ED-SP-GL-GBG:20} and references therein.
A basic online optimization setting
involves the minimization of time-varying convex loss functions,
resulting into Online Convex Programming (OCP). Typically, loss
objectives in OCP are functions of non-stationary stochastic
processes~\cite{MZ:03,EH:16}. Regret minimization aims to deal with non-stationarity by reducing the
difference between an optimal decision made with information in
hindsight, and one made as information is increasingly revealed.
Thus, several online algorithms and techniques are aimed at minimizing
various types of regret functions~\cite{AM-SS-AJ-AR:16,
  AJ-AR-SS-KS:15}. More recently, and with the aim of further reducing
the cost, regret-based OCP has integrated prediction models of loss
functions~\cite{AR-KS:13,ECH-RMW:15, NC-AA-AW-SB_LA:15,YL-GQ-NL:18}.
However, exact models of evolving loss functions may not be available,
while alternative data-based approximate models may require large
amounts of data that are hard to obtain.
This motivates the need of developing new learning algorithms for loss
functions that can employ finite data sets, while guaranteeing a
precise performance of the corresponding optimization.

\textit{Literature Review.}
Due to recent advances in Data Science and Machine Learning, the question of learning system models as well as distributional uncertainty from data is gaining significant attention.  From the early work on Systems Identification~\cite{LL:99}, Willem's Behavioral Theory and
fundamental lemma~\cite{JCW-PR-IM-BDM:05,TM-PR:17} have been recently  leveraged to learn linear, time-invariant system models in predictive control applications~\cite{TM-PR:17,CDP-PT:19,JC-JL-FD:19,JB-JK-MM-FA:20,AA-JC:20}. The aforementioned works rely on the use of Hankel system representations of the LTI system, and may be subject or not to additional uncertainty. In particular, the work~\cite{MN-MAM:22} leverages the behavioral theory to obtain sub-linear regret bounds for the online optimization of discrete-time unknown but deterministic linear systems. Other approaches to learn LTI systems from input-output data employ  concentration inequalities and finite samples, and include, for example,~\cite{SO-NO:19}, exploiting least squares and the Ho-Kalman algorithm,~\cite{AT-GJP:19}, using subspace identification techniques for LTI systems subject to unknown Gaussian disturbances, and~\cite{SF-NM-SS:19}, resorting to Lasso-like methods that exploit the sparse representation of LTI systems.

On the other hand, classical online optimization relies on Sample Averaging Approximation (SAA) (with bootstrap) to derive optimal value and/or policy approximations. However, SAA usually requires large amounts of data  to provide good approximations of the stochastic cost, which leads to non-robust solutions to unseen data.
In contrast, recent developments on measure-of-concentration results~\cite{NF-AG:15}  have lead to a new type of Distributionally Robust Optimization (DRO)~\cite{RG-AJK:16,PME-DK:17,SS-DK-PME:19}, which aims to bridge this gap. Particularly, the DRO framework enables finite-sample, performance-guaranteed optimization under distributional uncertainty~\cite{RG-AJK:16,PME-DK:17}, and paves the way to dealing with the control and estimation of system dynamics subject to distributional uncertainty.
Motivated by this, the works~\cite{DB-JC-SM:21,DB-JC-SM:20-acc} consider the time evolution of Wasserstein ambiguity sets and their updates under streaming data for estimation. However, the nominal dynamic constraints defined in these problems are assumed to be known, while in practice, these models also need to be identified. The previous work~\cite{DL-DF-SM:20-lcss} proposes a method for integrating the learning of an unknown and nominal parameterized system dynamics with Wasserstein ambiguity sets. These ambiguity sets are given by a parameter and control-dependent ambiguity ball center as well as a corresponding radius.  Taking this as a starting point, and motivated by the direct use of these ambiguity sets  in a  type of ``distributionally robust control'', here we further extend this setup in connection with online optimization problems.
Precisely, what distinguishes this work from other approaches is the focus  on learning the transition system dynamics itself via control-dependent ambiguity sets. The control method is derived from an online optimization method~\cite{AM-SS-AJ-AR:16}, and, therefore, it does not aim to calculate exactly an optimal control, but to find an approximate solution that leads to a low instantaneous regret function value w.r.t. standard, online and regret-based optimization problems.  Finally, this manuscript connects with the topic online optimization using decision-dependent distributions~\cite{DD-LX:23,KW-GB-EDA:22}, where the uncertainty distribution changes with the decision variable. As these problems are intractable, \cite{DD-LX:23,KW-GB-EDA:22} solve for alternative \textit{stable solutions}, or optimal solutions wrt to the distribution they induce. In addition to this, and while~\cite{DD-LX:23,KW-GB-EDA:22} can handle dynamic systems, a main difference with this work is that a dynamic system structure that is being learned is not exploited, which can help reduce uncertainty more effectively.

\textit{Statement of Contributions.}  In this
work, we propose a novel approach to solve a class of online optimization problems subject
  to distributionally uncertain dynamical systems.   Our end goal is to produce an online controller that results in bounded instantaneous regrets with high confidence.   Our proposed framework is
  unique in that it enables the online learning of the underlying nominal system,
  maintains online-problem tractability, and simultaneously provides
  finite-sample, probabilistic guarantee bounds on the resulting regret. This is achieved by
  considering a worst-case-system formulation that employs
  novel parameterized and control-dependent, Wasserstein ambiguity sets. Our learning method precisely consists of updating this ambiguity set.
  The proposed formulation is valid for a wide class
  of problems, including but not limited to 1) a class of optimal control
  problems subject to distributionally uncertain dynamical system, and 2)
  online resource allocation under distributional uncertainty.    To do this, we first obtain tractable problem reformulations for these two cases, which results in online and non-smooth convex problem optimizations.
  For each of these
  categories, and smoothed-out versions of these problems, we propose an online control algorithm dynamics,
which extends Nesterov's accelerated-gradient method. Adapting dissipativity theory, we prove  optimal first-order convergence rate for these algorithms under smoothness and convexity assumptions. This result is crucial to guarantee that the online controller can provide probabilistic guarantees on their regret bounds via the  control-dependent ambiguity set. We thus finish our work by quantifying these dynamic
regret bounds, and by
explicitly characterizing the effect of learning parameters with finite historical samples.

\section{Notations}\label{sec:Notation}
We denote by $\real^m$, $\realnonnegative^m$, $\integernonnegative^m$
and $\real^{m \times n}$ the $m$-dimensional real space, nonnegative
orthant, nonnegative integer-orthant space, and the space
of $m \times n$ matrices, respectively. The transpose of a column vector
$\vect{x} \in \real^m$ is
$\trans{\vect{x}}$, and $\vectorones{m}$ is a
  shorthand for $\trans{(1,\cdots,1)} \in \real^m$.  We index vectors
with subscripts, i.e., $\vect{x}_k \in \real^m$ with $k \in
\integernonnegative$, and given $\vect{x}\in \real^m$ we denote its
$\supscr{i}{th}$ component by $x_i$. We denote by $\Norm{\vect{x}}$
and $\Norm{\vect{x}}_{\infty}$ the $2$-norm and $\infty$-norm,
respectively.  The inner product of $\real^m$ is given as $\langle \vect{x}, \vect{y} \rangle:= \trans{\vect{x}}
\vect{y}$, $\vect{x}, \vect{y} \in \real^m$; thus,
$\Norm{\vect{x}}:=\sqrt{\langle \vect{x}, \vect{x} \rangle}$.
The gradient of a real-valued function $\map{\ell}{\real^m}{\real}$ is denoted as
$\nabla \ell(\vect{x})$ and $\nabla_{x}\ell(\vect{x})$ is the
partial derivative w.r.t.~$x$. In what follows,
$\dom(\ell):=\setdef{\vect{x} \in \real^m}{-\infty < \ell(\vect{x}) <+
  \infty}$.  A function $\map{\ell}{\dom({\ell})}{\real}$ is
$M$-strongly convex, if for any $\vect{y}, \vect{z} \in \dom({\ell})$
there exists $\vect{g}$ $\in \real^m$ such that
    $\ell(\vect{y}) \geq \ell(\vect{z}) + \trans{\vect{g}}
(\vect{y}-\vect{z}) + {M} \Norm{ \vect{y}-\vect{z} }^2 /{2}, \textrm{
  for some } M > 0$.      The function $\ell$ is convex if $M \geq 0$. We call a vector $\vect{g}$
a subgradient of $\ell$ at $\vect{z}$ and denote by $\partial
\ell(\vect{z})$ the subgradient set. If $\ell$ is
differentiable at $\vect{z}$, then $\partial \ell(\vect{z}) = \{
\nabla \ell(\vect{z}) \}$. Finally, the operation
$\map{\Pi_{\mathcal{U}}(\mathcal{X})}{\mathcal{X}}{\mathcal{U}}$
projects the set $\mathcal{X} \subseteq{\real^m}$ onto
$\mathcal{U} \subseteq{\real^m}$ under the Euclidean
norm. We write $\Pi_{\mathcal{\mathcal{U}}}(\vect{x}):=
\argmin_{\vect{z} } \Norm{\vect{x} -\vect{z} }^2/2 +
\chi_{\mathcal{U}}(\vect{z}) $, where $\vect{x} \in \mathcal{X}$, and
$\chi_{\mathcal{U}}(\vect{z})=0$ if $\vect{z} \in \mathcal{U}$,
otherwise $+\infty$.
Endow $\real^n$ with the Borel $\sigma$-algebra $\mathcal{B}$, and let $\mathcal{P}(\real^n)$ be the set of probability measures (or distributions) over $(\real^n,\mathcal{B})$.  The set  of probability distributions with bounded first moments is $\mathcal{M} =\{ \mathbb{Q} \in \mathcal{P}(\real^n)\,|\, \int_{\real^n} \|\vect{x}\| \text{d}\mathbb{Q} < +\infty\}$.
We use the Wasserstein
metric~\cite{KLV-RGS:58} to define a distance in
$\mathcal{M}$, and the dual version of the
$1$-Wasserstein metric $\map{d_W}{\mathcal{M} \times
  \mathcal{M}}{\realnonnegative}$, is defined by
  $d_W(\probQ_1,\probQ_2):= \sup_{f \in \mathcal{L}}\int
 {f(\vect{x})\text{d}\probQ_1} -\int
 {f(\vect{x})\text{d}\probQ_2} $,
  where $\mathcal{L}$ is the space of all Lipschitz functions with
 Lipschitz constant 1. We denote a closed Wasserstein ball of radius
 $\epsilon$ (also called an ambiguity set) centered at a distribution $\probP \in \mathcal{M}$ by
 $\mathbb{B}_{\epsilon}(\probP):=\setdef{\probQ \in
   \mathcal{M}}{d_W(\probP,\probQ) \le \epsilon}$.  The Dirac measure
 at $\vect{x}_0 \in \real^n$ is a distribution in $\mathcal{P}(\real^n)$ denoted by
 $\delta_{\{\vect{x}_0\}}$. Given $A \in
 \mathcal{B}$, we have $\delta_{\{\vect{x}_0\}}(A)=1$, if $\vect{x}_0 \in A$,
 otherwise $0$.
 A random vector $\vect{x} \in \real^m$ with probability distribution $\probQ$ is
   sub-Gaussian if there are positive constants $C, v$ such that $\probQ(\|\vect{x}\| > t) \le C \text{e}^{-v t^2}$.
 Equivalently, a zero-mean random vector $\vect{x} \in \real^n$ is sub-Gaussian if for any $a\in\real^n$ we have $\mathbb{E} \left[ \exp ( \trans{a} {\vect{x}}) \right] \leq \exp ( {\Norm{a}^2
     \nu^2}/{2})$ for some~$\nu$.
\section{Problem Statement, Motivation, and Approach based
    on  Ambiguity Set Learning}\label{sec:ProbStat}
We
  start by introducing a class of online optimization problems, where
the objective function is time-varying according to an
unknown dynamical system subject to an unknown
  disturbance.  Consider a dynamical system that evolves according to
unknown stochastic dynamics \begin{equation}
  \begin{aligned}
    \vect{x}_{t+1}=&f(t,\vect{x}_{t},\vect{u}_{t}) + \vect{w}_{t}, \;
    {\textrm{ from a given }}
    \vect{x}_0 \in \real^n,       \end{aligned} \label{eq:env}
\end{equation}
where $\vect{u}_t \in \mathcal{U} \subset \real^m$ is an online decision or control action at
time~$t$, $\map{f}{\realnonnegative \times \real^n \times \real^m}{\real^n}$ is a
measurable, but unknown state transition function, and $\vect{w}_{t} \in \real^n$, is an unknown, random, disturbance vector. Due to the Markov assumption, $\vect{x}_t \in
\real^n$ can be described by an unknown transition probability measure $\probP_{t|t-1}\in \mathcal{P}(\real^n)$, conditioned on the system state and control at time $t-1$. Denote by $\map{\ell}{\real^m  \times \real^n}{\real}$,
${(\vect{u},\vect{x})} \mapsto {\ell(\vect{u},\vect{x})}$ \textit{an
  a-priori} selected, measurable loss
function. Assume that $\mathcal{U}$ is compact,
  and we are interested in selecting $\vect{u}_t \in \mathcal{U}$ that minimizes the loss
\begin{equation*}
  \begin{aligned}
    \min\limits_{\vect{u}_t \in \mathcal{U}}  & \quad \left\{
      \mathbb{E}_{\probP_{t+1|t}} \left[ \ell(\vect{u}_t, \vect{x})\right]:=
      \int_{\real^n} \ell(\vect{u}_t, \vect{x}) \; \probP_{t+1|t}(d\vect{x})
    \right\} . \\
        \end{aligned}
\end{equation*} This objective value is inaccessible since the state distribution $\probP_{t+1|t}$ is unknown,
and its evolution is highly dependent on the
system, disturbance, and as well as on the decisions taken.
In this work, we aim to propose an effective online optimization and
learning algorithm which tracks the minimizers of the time-varying
objective function with low regret in high
probability. Thus, at each time $t$, we aim to
find $\vect{u}:=\vect{u}_t$ that minimizes the loss
in the immediate future at $t+1$ \begin{equation}
  \begin{aligned}
    \min\limits_{\vect{u} \in \mathcal{U}}  & \; \mathbb{E}_{\probP_{t+1|t}} \left[ \ell(\vect{u}, \vect{x})\right], \\
    \st \; & \;  \vect{x} \sim \probP_{t+1|t}, \textrm{ evolves according to~\eqref{eq:env}}.
  \end{aligned} \label{eq:P} \tag{P}
\end{equation}
This problem formulation is similar to a one-stage optimization problems with unknown system transitions~\cite{DB:12}.
The expectation operator with respect to $\probP_{t+1|t}$ is conditional on the historical realizations $\hat{\vect{x}}_k$, $k\leq t$, the adopted decisions $\hat{\vect{u}}_k$, $k\leq t-1$, the yet-to-be-learned unknown dynamical system $f$, and realizations $\hat{\vect{w}}_k$, $k\leq t-1$. We will identify $\probP_{t+1|t}(d\vect{x}) \equiv \probP_{t+1}(d\vect{x}|\vect{u}_t, \vect{x}_t=\hat{\vect{x}}_t, \vect{x}_k=\hat{\vect{x}}_k, \vect{u}_k=\hat{\vect{u}}_k, k \leq t-1)$ which, by the Markovian property, satisfies  $\probP_{t+1|t}(d\vect{x}) \equiv \probP_{t+1}(d\vect{x}|\vect{u}_t, \vect{x}_t=\hat{\vect{x}}_t)$.
At time $t$, let $\vect{u}^{\star}:=\vect{u}_t^{\star}$ denote an optimizer of Problem~\eqref{eq:P} and consider the instantaneous regret \begin{equation*}
  R_t:= \mathbb{E}_{\probP_{t+1|t}} \left[ \ell(\vect{u}, \vect{x})\right] -
  \mathbb{E}_{\probP_{t+1|t}} \left[ \ell(\vect{u}^{\star}, \vect{x})\right],
\end{equation*}
which is the loss incurred if the selected $\vect{u}$ is different from an optimal decision. Our goal will be to develop a robust online algorithm which ensures a probabilistic bound on the regret. That is, with high probability $\rho$, the regret $R_t$ is upper bounded by a sum of terms,
a first one depending on the initial condition $\vect{x}_0$; a second one depending on the instantaneous variation of the loss of~\eqref{eq:P}; and a third term related to how well the unknown system $f$ and the uncertainty are characterized; please see Theorem~\ref{thm:regret}. While the second and third terms are inherent to the system, the effect of the second one can be reduced by considering a predicted loss of the system~\cite{YL-GQ-NL:18}. In this work, we aim to bound the third term and minimize it by estimating the distribution $\probP_{t+1|t}$ via an ambiguity set of distributions. We will show that, as historical data are assimilated over time, this third term asymptotically decays to zero.
This is achieved under the following assumption \begin{assumption}[Independent and stationary sub-Gaussian
  distributions] \label{assump:subG} {\rm
  The vectors $\vect{w}_t  \in \real^n$, $t \in \integernonnegative$, are i.i.d. with $\vect{w}_t \sim \probQ$ and zero-mean $\sigma$ sub-Gaussian\footnote{That is, for all unit vector $\vect{v}$, we have  $\mathbb{E}[\text{e}^{\lambda \vect{v}^\top \vect{w}_t}] \le \text{e}^{\lambda^2 \sigma^2}/2 $,  $\forall \lambda \in \real$. Equivalently, $\probQ(\|\vect{w}_t\| > \lambda) \le  \text{e}^{-\lambda^2/(4\sigma^2)}$, $\forall \lambda$.}}.

\end{assumption}

\begin{remark}[On sub-Gaussian distributions] {\rm
Sub-Gaussian distributions include Gaussian random variables and all distributions of bounded support.                 } \end{remark}

  \noindent \textbf{Example 1 (Vehicle path planning and tracking):} A two-wheeled vehicle moves in an unknown
  2D environment. Assume that an
  accessible path-planner provides a control signal for the
  vehicle to track a desired reference trajectory under ideal
  conditions, see Fig.~\ref{fig:subG}.   Fig.~\ref{fig:path_plan} shows two examples where, first, the vehicle
  implements a series of lane changes, and, second, navigates   through a planned circular/loopy route. Since both the environment
  and dynamics are uncertain, exact tracking is rare.    Our goal
  is to learn the real-time road conditions,       and by solving the online
  problem~\eqref{eq:P},  derive a control signal that  enables path
  following minimizing the tracking error with high probability.
  \begin{figure}[tbp]
    \centering
  \includegraphics[width=0.4\textwidth]{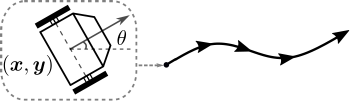}  \caption{\footnotesize A two-wheeled vehicle model with
      $(x,y)\in \real^2$ the position of the center and $\theta$ the
      direction.
    }  \label{fig:subG}
  \end{figure}
  \begin{figure}[tbp]
    \centering
  \includegraphics[width=0.125\textwidth]{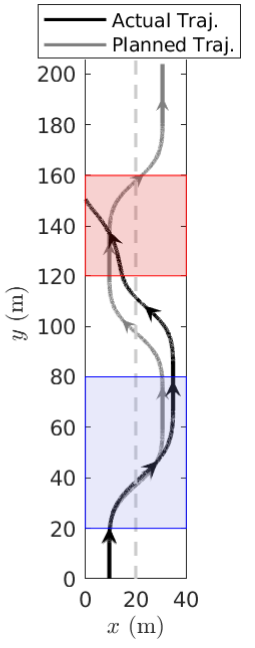}%
   \includegraphics[width=0.405\textwidth]{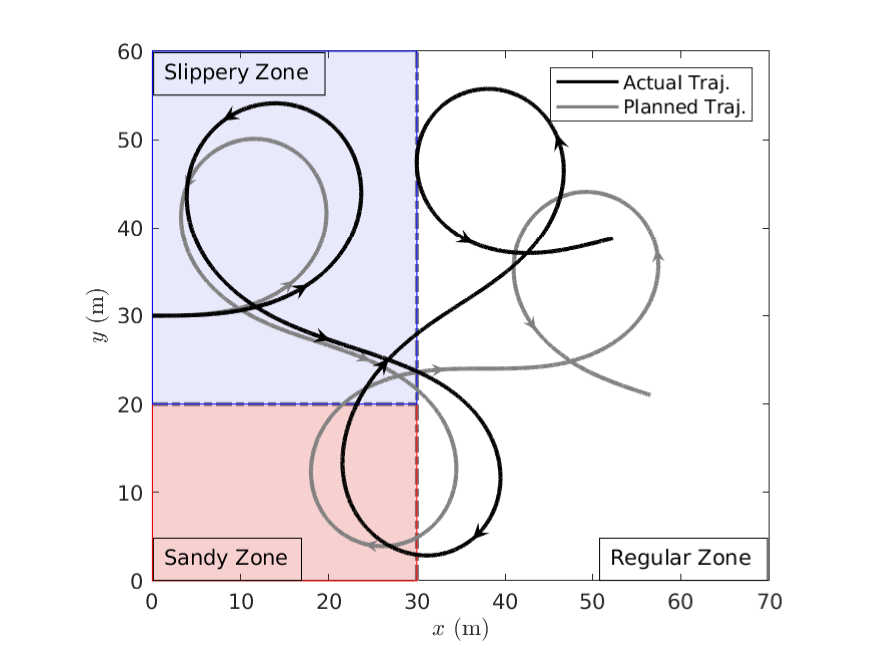}

   \caption{\footnotesize The (gray) planned trajectory and
      (black) actual system trajectory in various road zones, with the
      system state $\vect{x}=(x,y,\theta) \in \real^2 \times
      [-\pi,\pi)$. The red region indicates sandy zone while the blue
      region indicates the slippery zone. Due to unknown road
      conditions, the actual system trajectories deviate from planned
      trajectories.     }
      \label{fig:path_plan}
    \end{figure}

    \noindent \textbf{Example 2 (Online resource allocation in the
      stock market):} An agent aims to achieve a target
    profit of $r_0=130\%$ in a highly-fluctuating trading market. Thus, it actively allocates wealth to multiple risky assets
    while trying to balance resources among assets. As
    asset-prices are uncertain, modeling the return rate of each asset
    is specially challenging.         To solve this, an agent can aim
    to learn the real-time returns responsively, estimate the
    distributions of immediate returns, and then allocate wealth
    wisely to maximize the expected profit with high probability. This
    problem fits in the proposed formulation, resulting in online,
    balanced resource allocation with low regrets.
\subsection{Online Constructions of Ambiguity Sets
} \label{sec:ambiguitysets}
Our main approach to obtain a suitable control signal is
  based on learning a set of distributions or ambiguity set that
  characterizes system uncertainty.
More precisely, we employ the dynamic
ambiguity set $\mathcal{P}_{t+1}$ proposed
in~\cite{DL-DF-SM:20-lcss}. The set $\mathcal{P}_{t+1}$ contains a
class of distributions, which is, in high probability, large enough to
include the unknown $\probP_{t+1|t}$ under certain conditions. Thus, we
can use it to formulate a robust version of the problem at each time
instant $t$. Such characterization enables an online-tractable
reformulation of~\eqref{eq:P} later.
We summarize next the construction of these ambiguity sets $\mathcal{P}_{t+1}$. First, we assume the following on the unknown $f$.
\begin{assumption}[System parametrization] \label{assump:predictor} {\rm
Given $p\in \integerpositive$, the system $f$ can be expressed as
\begin{equation*}
  f(t,\vect{x},\vect{u})
  = \sum\limits_{i=1}^{p} \alpha^\star_i f^{(i)}(t,\vect{x},\vect{u}) ,
\end{equation*}
where $\vect{\alpha}^\star :=\trans{(\alpha_1^\star, \ldots, \alpha_{p}^\star)} \in \real^p$ is an unknown parameter, and $f^{(i)}:\realnonnegative \times \real^n \times \real^m \rightarrow \real^n $, $(t,\vect{x},\vect{u}) \mapsto
f^{(i)}(t,\vect{x},\vect{u})$, $i \in \until{p}$ is a set of $p$ linearly independent known basis functions or \textit{predictors} chosen \textit{a priori}.
}
\end{assumption}

Now, given arbitrary $(\vect{\alpha},\vect{u})$,  the set $\mathcal{P}_{t+1}$ is a Wasserstein ball centered at a parametric-dependent distribution $\hat{\probP}_{t+1|t}$ for each $t$; that is,
\begin{equation*}
  \mathcal{P}_{t+1}   := \mathbb{B}_{\hat{\epsilon}  }(\hat{\probP}_{t+1|t}   ) =\setdef{\probQ}{ d_W(\probQ, \hat{\probP}_{t+1|t}) \leq \hat{\epsilon} }.
\end{equation*}
 Here, $\hat{\epsilon}$ will be a time-varying function $\hat{\epsilon}\equiv\hat{\epsilon}(t, T, \beta, \vect{\alpha}, \vect{u})$ which depends on a number of $T$ measurements, and a confidence $\beta \in (0,1)$. More precisely,

\begin{equation} \label{eq:empdistn}
  \hat{\probP}_{t+1|t} (\vect{\alpha},\vect{u})
  := \frac{1}{T} \sum\limits_{k\in \mathcal{T}} \delta_{\{
 \sum_{i=1}^{p} \alpha_i \xi_k^{(i)}(\vect{\alpha},\vect{u})\}};
\end{equation}
see the footnote\footnote{
$\displaystyle {\xi}_{k}^{(i)}(\vect{\alpha},\vect{u}):= f^{(i)}({t},\hat{\vect{x}}_{t}, \vect{u}) + {\hat{\vect{x}}_{k+1}}/{(\trans{\vect{\alpha}} \vectorones{p})}-f^{(i)}({k},\hat{\vect{x}}_{k},\vect{u}_{k})$, with $\hat{\vect{x}}_{t}$, $\hat{\vect{x}}_{k+1}$, $\hat{\vect{x}}_{k}$ being the state measurements at time $t$, $k+1$, $k$ and $\vect{u}_{k}$ being the past input at $k$, $k \in \mathcal{T}$.}, where $\mathcal{T} = \{t-T,\dots, t\}$, for $t \ge T+1$. If $\vect{\alpha} = \vect{\alpha}^\star$, then $\sum_{i=1}^p \alpha_i {\xi}_{k}^{(i)}(\vect{\alpha},\vect{u}) \in \real^n$ provides an outcome $\vect{x}_{t+1}^{(k)} := f(t,\vect{x}_{t},\vect{u}) + \vect{w}_k $ $= \sum_{i=1}^p \alpha_i^\star f^{(i)}(t,\vect{x}_{t},\vect{u}) + \vect{w}_k$, for each $k$. For a general $\vect{\alpha} \approx \vect{\alpha}^\star$, the value $\sum_{i=1}^{p} \alpha_i \xi_k^{(i)}(\vect{\alpha},\vect{u})$ provides  ``approximated" outcomes  $\vect{x}_{t+1}^{(k)}$, for each $k = 1,\dots,T$.
Then, we claim the probabilistic guarantee of $\mathcal{P}_{t+1}$ by a selection of the parameter $\vect{\alpha}$ and $\hat{\epsilon}$ for any $\vect{u}$.
\begin{theorem}[Online probabilistic guarantee~{\cite[Application of Theorem 1]{DL-DF-SM:20-lcss}}] \label{thm:ambiguityset}
 Let Assumptions~\ref{assump:subG} and~\ref{assump:predictor} hold.
 For a given $T \in \integerpositive$, historical data $\{ \hat{\vect{x}}_k \}_{k\in\mathcal{T}}$ and $\{ \vect{u}_k \}_{k\in\mathcal{T}\setminus\{t\}}$, $\mathcal{T}= \{t-T,\dots, t\}$, we select $\hat{\probP}_{t+1|t}$ as in~\eqref{eq:empdistn} where $\vect{\alpha}$ is selected in~\cite[Theorem 2 (Learning of $\vect{\alpha}^{\star}$)]{DL-DF-SM:20-lcss}\footnote{In~\cite[Theorem 2]{DL-DF-SM:20-lcss}, the value $\vect{d}$ plays the role of $\vect{u}$ in this work.}. Then, for given $\vect{u}$ and a confidence-related value
$\beta\in(0,1)$, a radius $\hat{\epsilon}:=\hat{\epsilon}(t,T,\beta,\vect{\alpha},\vect{u})$ can be chosen such that
\begin{equation}
  \begin{aligned}
    &{\Prob}\left( \probP_{t+1|t} \in
      \mathcal{P}_{t+1}     \right) \geq \rho(t) .
\end{aligned} \label{eq:modgua}
\end{equation}
Here, the left-hand-side expression is a shorthand for the probability of the event $ \{(\vect{x}_{t+1}^{(1)}, \dots, \vect{x}_{t+1}^{(T)}) \in \real^n \times \dots \times \real^n \,|\, \probP_{t+1|t} \in \mathbb{B}_{\hat{\epsilon}}(\hat{\probP}_{t+1|t})\}$ and $\Prob:=\probP_{t+1|t}^T$ denotes the probability measure defined
on the $T$-fold product of $\probP_{t+1|t}$, which evaluates the probability that the selection of samples define an ambiguity ball which contains the true distribution. In particular, the confidence value is
\begin{equation*}
     \rho(t) :=    \left( 1- \beta \right) \left(   1- \exp \left( - \frac{ ( \gamma^2 - \sqrt{2}c \gamma) T }{2 \sqrt{2} \left( c \gamma  + \sqrt{2} c^2 \right)} \right)  \right),
   \end{equation*}
   where $c$ is a data-dependent positive constant and $\gamma> \sqrt{2}c$ is a user selected parameter. Further, the radius is
\begin{equation} \label{eq:epsilon} \small
\begin{aligned}
      \hat{\epsilon}  :=& \sqrt{ \frac{2n M\sigma^2}{T} \ln(\frac{1}{\beta})} + C_1 T^{- {1}/{\max \{ n, 2  \} }} + \gamma H(t,T,{\vect{u}}) , \\
  \end{aligned}
\end{equation}
where $M$ and $C_1$ are positive constants, and
\begin{equation*}
  H(t,T,{\vect{u}}):= \frac{1}{T} \sum\limits_{i=1}^{p} \sum\limits_{k\in \mathcal{T}} \Norm{ f^{(i)}(k,\hat{\vect{x}}_{k},{\vect{u}}_{k}) - f^{(i)}(t,\hat{\vect{x}}_{t},{\vect{u}}) },
\end{equation*}
which bounds the variation of predicted system trajectories.

\end{theorem}

\textit{Idea of the Proof.} The probabilistic guarantees~\eqref{eq:modgua} are a consequence of Lemma 1, Theorem 1, Theorem 2 and Eqn.~(7) in~\cite{DL-DF-SM:20-lcss} with Assumptions~\ref{assump:subG} and~\ref{assump:predictor}. Precisely, we achieve this by upper bounding the metric $d_W({\probP}_{t+1|t}, \hat{\probP}_{t+1|t}(\vect{\alpha},\vect{u}))$ using $d_W({\probP}_{t+1|t}, \hat{\probP}_{t+1|t}(\vect{\alpha}^\star,\vect{u}))$ plus $d_W(\hat{\probP}_{t+1|t}(\vect{\alpha}^\star,\vect{u}),\hat{\probP}_{t+1|t}(\vect{\alpha},\vect{u}))$.
Then, the first distance is handled via ~\cite[Lemma 1]{DL-DF-SM:20-lcss} using standard measure of concentration results\footnote{Lemma~1 in~\cite{DL-DF-SM:20-lcss} makes use of a stronger Assumption~\ref{assump:subG}, which requires $\vect{w}_k$ to be white. However, this can be relaxed to the current assumption by multiplying the upper bound in the lemma with a constant $M>0$ associated with noise whitening via an appropriate linear transformation.}, contributing to the first two terms of the radius $\hat{\epsilon}$ in~\eqref{eq:epsilon}.
Next, the second distance $d_W(\hat{\probP}_{t+1|t}(\vect{\alpha}^\star,\vect{u}),\hat{\probP}_{t+1|t}(\vect{\alpha},\vect{u}))$ can be bounded in terms of the difference $\|\vect{\alpha} - \vect{\alpha}^\star\|$ via ~\cite[Theorem 1]{DL-DF-SM:20-lcss}, contributing to the third term in $\hat{\epsilon}$. Notice that the third term depends on Assumption~\ref{assump:predictor} and the selected parameter $\gamma$ which relies on the selection of $\vect{\alpha}$ via~\cite[Theorem 2 (Learning of $\vect{\alpha}^{\star}$)]{DL-DF-SM:20-lcss}. The confidence value $\rho(t)$ is achieved by Assumption~\ref{assump:subG} applying to the same procedure as in~\cite[Theorem 2]{DL-DF-SM:20-lcss}, which essentially bounds $\Norm{\vect{\alpha} -\vect{\alpha}^{\star}}_{\infty}$ in probability. Precisely, by Assumption~\ref{assump:subG}, we have $\probQ(\|\vect{w}_t\|_{\infty} > \eta) \le  \text{e}^{-\eta^2/(4\sigma^2)}$,  $\forall \; \eta$, resulting in $\mathbb{E}\left[ \Norm{\vect{w}_t}_{\infty}^{l} \right] \leq 2^{\frac{l}{2}-1}\sigma^{l} l^{\frac{l}{2}+1}, \;\; \forall \, l \in
  \integernonnegative$, analogous to~~\cite[Lemma 2]{DL-DF-SM:20-lcss-arxiv}. Then, with the proof similar to~\cite[Theorem IV.2]{DL-DF-SM:20-lcss-arxiv}, we achieve
 \begin{equation*}
{\small
  \Prob \left( \frac{1}{T} \sum\limits_{k \in \mathcal{T}} \left(    \Norm{\vect{w}_k}_{\infty}  \right) \geq \gamma  \right)
\leq \exp \left( -\gamma \lambda +  \frac{ T \sqrt{2}
 \sigma e \lambda }{ 2T- 2 \sqrt{2} \sigma e \lambda}  \right) .
}
\end{equation*}
By selecting
\begin{equation*} \small
\lambda=
\begin{cases}
  \frac{T}{2 \sqrt{2}
  \sigma e} -\frac{T}{2 \gamma}, &\textrm{ if } \gamma \geq \sqrt{2}
   \sigma  e , \\
    0, &\textrm{ if } \gamma < \sqrt{2}
     \sigma e ,
\end{cases}
\end{equation*}
we follow the proof~\cite[Theorem IV.2]{DL-DF-SM:20-lcss-arxiv} to achieve
\begin{equation*}
  \begin{aligned}
  & \Prob \left( \frac{1}{T} \sum\limits_{k \in \mathcal{T}} \left(    \Norm{\vect{w}_k}_{\infty}  \right) \geq \gamma  \right)    \\
  & \hspace{10ex}  \leq \begin{cases}
    \exp \left( - \frac{ T (\gamma^2 - \sqrt{2}\sigma e \gamma) }{2 \sqrt{2} \sigma e (\gamma + \sqrt{2} \sigma e)  } \right) ,  &\textrm{ if } \gamma \geq
        \sqrt{2} \sigma e , \\
    1 , &\textrm{ if } \gamma <
     \sqrt{2} \sigma e .
\end{cases}
  \end{aligned}
\end{equation*}
By bound propagation, we have
\begin{equation*} \small
  \begin{aligned}
  &  \Prob \left( \Norm{\vect{\alpha} - \vect{\alpha}^{\star}}_{\infty}  \leq \gamma \right)  \geq
   1- \exp \left( - \frac{ ( \gamma^2 - \sqrt{2}c \gamma) T }{2 \sqrt{2} \left( c \gamma  + \sqrt{2} c^2 \right)} \right) ,
  \end{aligned}
\end{equation*}
with $\gamma> \sqrt{2} c$ and $c$ is selected as in~\cite[Theorem 2]{DL-DF-SM:20-lcss}.
Finally, the combination of all the above considerations complete the proof.  \qed

Theorem~\ref{thm:ambiguityset} provides a methodology to construct online ambiguity sets with guarantees in probability. In general, $\rho(t)$ is strictly smaller than 1 unless there is a way of making $\vect{\alpha}(t) \rightarrow \vect{\alpha}^\star$. This is implemented in~\cite{DL-DF-SM:20-lcss} via an online learning algorithm which leads to $\rho(t) \rightarrow 1- \beta$ via Eqn.~(7) in the same work. Notice how these constructions are related to the decision variable $\vect{u}$ and,
in the following, we leverage the probabilistic
characterization
$\mathcal{P}_{t+1}:=\mathcal{P}_{t+1}(\vect{\alpha},\vect{u})$ of the
distribution $\probP_{t+1|t}$ for solutions
to~\eqref{eq:P}.

\section{A Tractable Problem Reformulation
  and Its Specialization to Two  Problem Classes}
 \label{sec:tractable}
 In this section, we start by describing how to deal with
   the unknown $\probP_{t+1|t}$ in Problem~\eqref{eq:P}, via
   ambiguity sets, which results in~\eqref{eq:P1}.
 By doing this, the solution of~\eqref{eq:P1} provides
 guarantees on the performance
 of~\eqref{eq:P}. Unfortunately, this results into an
   online intractable problem. Thus, we find a tractable
 reformulation~\eqref{eq:P2} which is equivalent to~\eqref{eq:P1}
 under certain conditions. After this, we focus the rest of
   our work on two problem sub-classes, which allows us to present and
   analyze the online algorithms for these problems in the following
   section. Formally, let us consider
\begin{equation}
  \begin{aligned}
    \min\limits_{\vect{u} \in \mathcal{U}}  \sup\limits_{\probQ \in
      \mathcal{P}_{t+1}(\vect{\alpha},\vect{u})} & \;
    \mathbb{E}_{\probQ} \left[ \ell(\vect{u}, \vect{x})\right],         \end{aligned} \label{eq:P1} \tag{P1}
\end{equation}
where, for a fixed $\vect{\alpha}:=\vect{\alpha}_t$ and
${\vect{u}}:={\vect{u}_t} \in \mathcal{U}$, it holds that
$\probP_{t+1|t} \in \mathcal{P}_{t+1}(\vect{\alpha},\vect{u})$ with high
probability. This results in
\begin{equation*}
  \begin{aligned}
    & \Prob \left( \mathbb{E}_{\probP_{t+1|t}} \left[ \ell( {\vect{u}},
        \vect{x})\right] \leq \sup\limits_{\probQ \in
        \mathcal{P}_{t+1}       }\; \mathbb{E}_{\probQ} \left[ \ell( {\vect{u}},
        \vect{x})\right] \right) \geq \rho (t).
\end{aligned}   \end{equation*}
Observe that, the probability measure $\Prob$ and the bound $\rho(t)$ coincides
with that in~\eqref{eq:modgua} and notice how the value
$\rho(t)$ changes for various data-set sizes $T$
in Theorem~\ref{thm:ambiguityset}.

The solution $\vect{u}$ and the objective value of~\eqref{eq:P1}
ensure that, when we select $\vect{u}$ to be the decision
for~\eqref{eq:P}, the expected loss of~\eqref{eq:P}
  is no worse than that from~\eqref{eq:P1} with high
  probability.
      The
formulation~\eqref{eq:P1} still requires expensive online computations
due to its semi-infinite inner optimization problem.
Thus, we propose an equivalent reformulation of~\eqref{eq:P1} for a
class of loss functions as in the following assumption.

\begin{assumption}[Lipschitz loss
  functions] \label{assump:cvxloss} {\rm Consider the loss function
    $\map{\ell}{\real^m \times \real^n}{\real}$,
    ${(\vect{u},\vect{x})} \mapsto {\ell(\vect{u},\vect{x})}$. There
    exists a Lipschitz function $\map{L}{\real^m}{\realnonnegative}$
    such that for each $\vect{u} \in \real^m$, it holds that
    $\Norm{\ell(\vect{u}, \vect{x}) - \ell(\vect{u}, \vect{y})} \leq L
    (\vect{u}) \Norm{\vect{x} - \vect{y}}$ for any $\vect{x}, \vect{y}
    \in \real^n$.  }
\end{assumption}

With this, we obtain the following upper bound:
\begin{lemma}[An upper bound of~\eqref{eq:P1}]
  Let Assumption~\ref{assump:cvxloss} hold. Then, for each $\vect{u}$,
  $\vect{\alpha}$, $\beta$, $T$ and $t$, we have \begin{equation*}
  \begin{aligned}
    \sup\limits_{\probQ \in \mathcal{P}_{t+1}(\vect{\alpha},\vect{u})} & \; \mathbb{E}_{\probQ} \left[ \ell(\vect{u}, \vect{x})\right]  \\
    & \hspace{0ex} \leq
    \mathbb{E}_{\hat{\probP}_{t+1|t}(\vect{\alpha},\vect{u})} \left[
      \ell(\vect{u}, \vect{x})\right] +
    \hat{\epsilon}(t,T,\beta,\vect{\alpha},\vect{u}) L(\vect{u}),
  \end{aligned}
\end{equation*}
where the empirical distribution
$\hat{\probP}_{t+1|t}(\vect{\alpha},\vect{u})$ and scalar
$\hat{\epsilon}(t,T,\beta,\vect{\alpha},\vect{u})$ are described as in Section~\ref{sec:ambiguitysets}.
\label{lemma:Reform}
\end{lemma}
Hereafter, see the appendix for all proofs.

Next, we claim that the upper bound in Lemma~\ref{lemma:Reform} is
tight if the following assumption holds.
\begin{assumption}[Convex and gradient-accessible
  functions] \label{assump:gradient} {\rm The loss function $\ell$ is
    convex in $\vect{x}$ for each
      $\vect{u}$. Further, for each time
    $t$ with given $\vect{u}(=\vect{u}_t) \in \mathcal{U}$ and $\vect{\alpha}(=\vect{\alpha_t}) \in
    \real^p$, there is a system prediction $\sum_{i=1}^{p}
    \alpha_i \xi_k^{(i)}(\vect{\alpha},\vect{u})$ for some     $k \in \mathcal{T}$ such that $\nabla_{\vect{x}} \ell$ exists and
    $L(\vect{u})$ is equal to $\Norm{\nabla_{\vect{x}}
      \ell}$     at $(\vect{u}, \sum_{i=1}^{p} \alpha_i
    \xi_k^{(i)}(\vect{\alpha},\vect{u}))$.   }
          \end{assumption}
The above statement enables the following theorem.
\begin{theorem}[Equivalent reformulation
  of~\eqref{eq:P1}] \label{thm:Reform}
  Let
  Assumptions~\ref{assump:cvxloss}
 and~\ref{assump:gradient} hold. Let
  $\Xi_{t+1}$ denote the support of the distribution
  $\probP_{t+1|t}$. Then, if
    $\Xi_{t+1}=\real^n$,~\eqref{eq:P1} is equivalent to the
  following   problem
\begin{equation}
  \begin{aligned}
                        \min\limits_{\vect{u} \in \mathcal{U} } \;
    \mathbb{E}_{\hat{\probP}_{t+1|t}(\vect{\alpha},\vect{u})} \left[
      \ell(\vect{u}, \vect{x})\right] +
    \hat{\epsilon}(t,T,\beta,\vect{\alpha},\vect{u}) L(\vect{u}).
  \end{aligned} \label{eq:P2} \tag{P2}
\end{equation}
\end{theorem}

  \begin{remark}[Effects of Assumptions~\ref{assump:cvxloss}
    and~\ref{assump:gradient}] {\rm We note that
      Assumption~\ref{assump:cvxloss} on the Lipschitz requirement of
      loss function is mild. In fact, many engineering problems take
      state values in a compact set, which then only requires the loss
      $\ell$ to be continuous. Assumption~\ref{assump:gradient}
      essentially requires accessible partial gradients (in
      $\vect{x}$) of loss functions
      $\ell$.       For simple loss functions $\ell$, e.g. linear, quadratic, etc,
      its partial gradient can be readily evaluated. Notice that when
      Assumption~\ref{assump:gradient} fails, Problem~\eqref{eq:P2}
      still serves as a relaxation problem of~\eqref{eq:P1}, providing
      a solution with a valid upper bound.}
\end{remark}

Notice that the tractability of solutions to~\eqref{eq:P2} now depend
on: 1) the choice of the loss function $\ell$ and the associated
Lipschitz function $L$, and 2) the decision space $\mathcal{U}$.  To
be able to further analyze~\eqref{eq:P2} and further
  evaluate Assumption~\ref{assump:gradient} on gradient-accessible functions,
    we
will impose further structure on the system as follows:
\begin{assumption}[Locally Lipschitz, control-affine system and
  basis functions]
  The system $f$ is locally Lipschitz in $(t,\vect{x},\vect{u})$
  and affine in $\vect{u}$, i.e.,
\begin{equation*}
  f(t, \vect{x},\vect{u}):= f_1(t,\vect{x}) +  f_2(t,\vect{x}) \vect{u},
\end{equation*}
for some unknown $\map{f_1}{\realnonnegative \times
  \real^n}{\real^n}$, $\map{f_2}{\realnonnegative \times
  \real^n}{\real^{n\times m} }$, $\vect{u} \in \mathcal{U}$ and $t \in
\integernonnegative$. Similarly, for each $i\in \until{p}$, the
 basis function $f^{(i)}$ is selected to be
\begin{equation*}
  f^{(i)}(t, \vect{x},\vect{u}):= f^{(i)}_1(t,\vect{x}) +  f^{(i)}_2(t,\vect{x}) \vect{u},
\end{equation*}
for some  known locally Lipschitz functions $f^{(i)}_1$ and
$f^{(i)}_2$.
\label{assump:affine}
\end{assumption}
\begin{assumption}[Convex decision oracle]
  \label{assump:cvxu}   The set $\mathcal{U}$ is convex and compact. Furthermore, the projection
  operation of $\vect{u}\in \real^m$ onto $\mathcal{U}$,
  $\Pi_{\mathcal{U}}(\vect{u})$, admits $O(1)$ computation complexity.
\end{assumption}
\begin{comment}
\begin{remark}[$\mathcal{U}$ examples]
  Examples of $\mathcal{U}$ include the following: 1) the non-negative
  orthant as $\setdef{\vect{u}}{\vect{u} \geq \vectorzeros{m}}$, 2) an
  $m$-cell as $\setdef{\vect{u}}{\underline{\vect{u}} \leq \vect{u}
    \leq \overline{\vect{u}} }$, for some constant vectors
  $\underline{\vect{u}}$ and $\overline{\vect{u}}$, 3) a unit simplex
  as $\setdef{\vect{u}}{\trans{\vect{u}} \vectorones{m} =1 , \;
    \vect{u} \geq \vectorzeros{m} }$, or 4) a ball $\setdef{\vect{u}}{
    \Norm{\vect{u}} \leq 1 }$.
\end{remark}
\end{comment}

For simplicity of the discussion, we rewrite~\eqref{eq:P2} as
\begin{equation*}
  \begin{aligned}
            \min\limits_{\vect{u} \in \mathcal{U} } \; G(t,\vect{u} ) :=  G(t,\vect{u} | \ell, L, T, \beta, \vect{\alpha}, \hat{\probP}_{t+1|t}, \hat{\epsilon}  ) ,
  \end{aligned}
\end{equation*}
where $G$ represents the objective function of~\eqref{eq:P2}, depending on  variables $\ell$, $L$,  $\beta$, $\vect{\alpha}$, and $\mathcal{P}_{t+1}$,  which are kept fixed in the optimization. Then, Assumption~\ref{assump:affine} allows an explicit expression of $G$ w.r.t. $\vect{u}:=\vect{u}_t$ and
Assumption~\ref{assump:cvxu} characterizes the convex feasible set of~\eqref{eq:P2}. Note that $G(t,\vect{u})$ is locally Lipschitz in $t$.\footnote{This can be verified by the local Lipschitz condition on $f^{(i)}$, $\ell$, and finite composition of local Lipschitz functions are locally Lipschitz. }

  In the following, we consider two classes of general problems
in the form of~\eqref{eq:P2}: 1) an optimal control problem under the
uncertainty; 2) an online resource allocation problem with a
switch. These problems leverage the probabilistic characterization of
the system and common loss functions $\ell$. Then, we
propose an online algorithm to achieve tractable solutions with a
probabilistic regret bound in the next section.

\noindent \textbf{Problem 1: (Optimal control under uncertainty)} We
consider a problem in form~\eqref{eq:P}, where the system is unknown and is to be optimally controlled. In particular, we employ the
following separable loss function   \begin{equation*}
  \ell(\vect{u}, \vect{x}):= \ell_1(\vect{u}) +
  \ell_2(\vect{x}), \quad \map{\ell_1}{\real^m}{\real},\; \map{\ell_2}{\real^n}{\real},
\end{equation*}
with $\ell_1$ the cost for the immediate control and $\ell_2$ the
optimal cost-to-go function. We assume that both $\ell_1$ and $\ell_2$
are convex, and in addition, $\ell_2$ is Lipschitz continuous with a
constant $\Lip(\ell_2) \in\realnonnegative$, resulting in
$L(\vect{u}) \equiv \Lip(\ell_2)$. Then, by selecting the
ambiguity radius $\hat{\epsilon}$ and center $\hat{\probP}_{t+1|t}$ of
$\mathcal{P}_{t+1}$ as in Section~\ref{sec:ambiguitysets},
the objective
function  of~\eqref{eq:P2} becomes \begin{equation*}
  \begin{aligned}
    &G(t,\vect{u} )= \ell_1(\vect{u}) + \frac{1}{T} \sum\limits_{k\in
      \mathcal{T}} \ell_2 ( \vect{p}_{k,t}
    )  \\
    & \hspace{2.2cm} +\Lip(\ell_2) \epsilon     + \frac{\gamma \Lip(\ell_2) }{T} \sum\limits_{i=1}^{p}
    \sum\limits_{k\in \mathcal{T}} \Norm{\vect{H}_k^{(i)} },
    \end{aligned}
  \end{equation*}
  where $\vect{p}_{k,t}$, $\vect{H}_k^{(i)} \in \real^n $ are affine in $\vect{u}$, for
  each $i$, $k$, as
  \begin{equation*}
    \begin{aligned}
      & \vect{p}_{k,t}:= \sum\limits_{i=1}^{p} \alpha_i
      \left(  f^{(i)}_1({t},\hat{\vect{x}}_{t})  -f^{(i)}({k},\hat{\vect{x}}_{k},\vect{u}_{k}) \right) \\
      & \hspace{2.8cm} + \hat{\vect{x}}_{k+1} + \left(
        \sum\limits_{i=1}^{p} \alpha_i
        f^{(i)}_2({t},\hat{\vect{x}}_{t}) \right) \vect{u},
    \end{aligned}
  \end{equation*}
  \begin{equation*}
    \vect{H}_k^{(i)}(\vect{u}):=
    f^{(i)}(k,\hat{\vect{x}}_{k},\vect{u}_{k})
    - f^{(i)}_1(t,\hat{\vect{x}}_{t}) -f^{(i)}_2(t,\hat{\vect{x}}_{t}) {\vect{u}},
  \end{equation*}
  and parameters $\vect{\alpha} \in\real^p$, $\epsilon \in \realnonnegative$ and $\gamma \in \realnonnegative$ are selected as in~\cite[Section IV]{DL-DF-SM:20-lcss}.
  Intuitively, $\vect{p}_{k,t}$ is the $\supscr{k}{th}$ projected outcome of the random variable $\vect{x}_{t+1}$ and $\vect{H}_k^{(i)}$ quantifies the variation of predictor $f^{(i)}$ with respect to its previous $\supscr{k}{th}$ value.
  Notice that the objective function $G$ is convex in $\vect{u}$
and therefore online problems~\eqref{eq:P2} are tractable. In
addition, if $\ell_2$ has a constant gradient almost everywhere,
then Assumption~\ref{assump:gradient} on accessible gradients holds and~\eqref{eq:P2} is
equivalent to~\eqref{eq:P1}.

\noindent \textbf{Problem 2: (Online resource allocation)} We consider
an online resource allocation problem with a
switch, where a decision maker aims to make online resource allocation decisions in an uncertain environment.
This problem is in form~\eqref{eq:P} and its objective
is \begin{equation*}
  \ell(\vect{u},\vect{x})= \max \{ 0, 1-  \langle \vect{u},
  \phi(\vect{x})  \rangle   \}, \quad \map{\phi}{\real^n}{\real^m},
\end{equation*}
where $\phi$ is an affine feature map selected in advance. The
decision maker updates the decision $\vect{u}$ online when $\langle
\vect{u}, \phi(\vect{x}) \rangle < 1$, otherwise switches
off. Notice that this type of objective functions appears in many classification problems.
  In particular, we assume that the system $f$ is
independent from the allocation variable, i.e., $f_2 \equiv
0$. See Section~\ref{subsec:resource} for a more
  explicit problem formulation involving resource allocation with an assignment switch.

Then,       problem~\eqref{eq:P2} has the objective function
\begin{equation*}
  \begin{aligned}
    &G(t,\vect{u} )= \frac{1}{T} \sum\limits_{k\in \mathcal{T}} \max
    \{ 0, 1- \langle \vect{u}, \phi(\vect{p}_{k,t}) \rangle \} +
    q_{t} L(\vect{u}),
  \end{aligned}
\end{equation*}
where time-dependent parameters $\vect{p}_{k,t}\in \real^n$, $q_{t}\in \real$ are
\begin{equation*}
  \begin{aligned}
    & \vect{p}_{k,t} = \hat{\vect{x}}_{k+1} + \sum\limits_{i=1}^{p}
    \alpha_i \left( f^{(i)}_1({t},\hat{\vect{x}}_{t}) -
      f^{(i)}_1({k},\hat{\vect{x}}_{k}) \right), \; \forall \; k, \; t, \\
    & q_{t}= \epsilon     + \frac{\gamma }{T} \sum\limits_{i=1}^{p} \sum\limits_{k\in
      \mathcal{T}} \Norm{ f^{(i)}_1(k,\hat{\vect{x}}_{k}) -
      f^{(i)}_1(t,\hat{\vect{x}}_{t}) } , \; \forall \; t,
  \end{aligned}
\end{equation*}
with $\vect{\alpha} \in\real^p$, $\epsilon \in \realnonnegative$ and $\gamma \in \realnonnegative$ as in~\cite[Section IV]{DL-DF-SM:20-lcss}.
We characterize the function $L(\vect{u})$ by subgradients of the loss function $\ell$. \begin{lemma}[Quantification of $L$] \label{lemma:Lip}     Consider $\ell(\vect{u},
  \vect{x}):= \max \{ 0, 1- \langle \vect{u}, \phi(\vect{x}) \rangle
  \}$, where $\phi(\vect{x})$ is differentiable in $\vect{x}$.  Then,
  the function $L(\vect{u})$ is
\begin{equation*}
  L(u)= \sup
  \limits_{\vect{g} \in \partial_{\vect{x}} \ell(\vect{u},\vect{x}), \; \vect{x} \in \real^n } \Norm{\vect{g}},
\end{equation*}
where the set $\partial_{\vect{x}} \ell(\vect{u},\vect{x})$ contains
all the subgradients of $\ell$ at $\vect{x}$, given any $\vect{u}$ in
advance, i.e.,
\begin{equation*}
  \begin{aligned}
    \partial_{\vect{x}} \ell(\vect{u},\vect{x}):= h(\vect{x},
    \vect{u}) \cdot \frac{\partial \phi}{\partial \vect{x}}(\vect{x})
    \vect{u},
  \end{aligned}
\end{equation*}
where
\begin{equation*}
  h(\vect{x}, \vect{u})=\begin{cases} -1,
    &\textrm{ if } \langle \vect{u}, \phi(\vect{x})  \rangle <1 \\  0,
    &\textrm{ if } \langle \vect{u}, \phi(\vect{x})  \rangle >1 \\ [-1,0], &\textrm{ otherwise } \end{cases}.
\end{equation*}
In particular, if $\phi(\vect{x}):=J \vect{x}$ for some matrix $J \in \real^{m\times n}$,
then
$  L(u)=  \Norm{ \trans{J} \vect{u}} $. If $\vect{x}$ is contained in a compact set $X$, then
 $ L(u)=  \Lip(\phi) \Norm{\vect{u}}, $where $\Lip(\phi) \in \realnonnegative$ is the Lipschitz constant of $\phi$ on $X$. \end{lemma}

Lemma~\ref{lemma:Lip} indicates that, given a properly selected
feature mapping $\phi$, the objective $G$ is convex in $\vect{u}$ and therefore
online problems~\eqref{eq:P2} are convex and tractable. In addition,
if $\phi$ is a linear map almost everywhere, then
Assumption~\ref{assump:gradient} on accessible gradients holds and~\eqref{eq:P2} is equivalent
to~\eqref{eq:P1}.
\section{Online Algorithms} \label{sec:alg}
Online convex problems~\eqref{eq:P2} are non-smooth due to the
normed regularization terms in $G$. To achieve fast, online solutions,
we propose a two-step procedure.  First, we follow~\cite{AB-MT:12,YN:05} to obtain a smooth version of~\eqref{eq:P2},
called~\eqref{eq:P2smooth}. Then, we extend
the \textit{Nesterov's accelerated-gradient} method~\cite{YN:13}---known to achieve an optimal
  first-order convergence rate for smooth and offline convex
  problems---to solve the problem~\eqref{eq:P2smooth}.
Finally, we quantify the dynamic regret~\cite{MZ:03}
of online decisions w.r.t. solutions of~\eqref{eq:P1} in probability. \\
\noindent \textbf{Step 1: (Smooth approximation of~\eqref{eq:P2})} To
simplify the discussion, let us use the generic notation
$\map{F}{\mathcal{U}}{\real}$ for a convex and potentially non-smooth
function, which can represent any particular component of the
objective function $G(t,\vect{u})$ of~\eqref{eq:P2} at time $t$.

\begin{definition}[Smoothable function~\cite{AB-MT:12}] \label{def:smooth} We call a
  convex function $F(\vect{u})$ \textit{smoothable} on $\mathcal{U}$
  if there exists $a>0$ such that, for every $\mu>0$, there is a
  continuously differentiable convex function
  $\map{F_{\mu}}{\mathcal{U}}{\real}$ satisfying \\
    \noindent (1) $F_{\mu}(\vect{u})
  \leq F(\vect{u}) \leq F_{\mu}(\vect{u})+ a \mu$, for all $\vect{u}\in \mathcal{U}$. \\
  \noindent (2) There exists $b>0$ such
  that $F_{\mu}$ has a Lipschitz gradient over $\mathcal{U}$ with Lipschitz constant $b/ \mu$, i.e., \\
  \begin{equation*}
    \Norm{ \nabla F_{\mu}(\vect{u}_1) -\nabla F_{\mu}(\vect{u}_2)} \leq \frac{b}{\mu} \Norm{\vect{u}_1 - \vect{u}_2},
        \; \forall \; \vect{u}_1, \vect{u}_2 \in \mathcal{U}.
  \end{equation*}
\end{definition}
To obtain a smooth approximation $F_{\mu}$ of $F$, we follow the
\textit{Moreau proximal approximation} technique~\cite{AB-MT:12},
described as in the following lemma.

\begin{lemma}[Moreau-Yosida approximation] \label{lemma:moreau} Given
  a convex function $\map{F}{\mathcal{U}}{\real}$ and any $\mu>0$, let
  us denote by   $\partial F(\vect{u})$ the set of subgradients of $F$ at $\vect{u}$,
  respectively. Let $D:=\sup_{ g \in \partial F(\vect{u}), \vect{u}\in
    \mathcal{U}} \Norm{g}^2< +\infty$. Then, $F$
  is smoothable with parameters $(a,b):= (D/2, 1)$, where the smoothed
  version   $\map{F_{\mu}}{\mathcal{U}}{\real}$ is the Moreau approximation:
\begin{equation*}
  F_{\mu}(\vect{u}):= \inf\limits_{\vect{z} \in \mathcal{U}}
  \left\{ F(\vect{z}) + \frac{1}{2\mu} \Norm{\vect{z}- \vect{u} }^2
  \right\}, \; \vect{u} \in \mathcal{U}.
\end{equation*}
In addition, if $F$ is $M$-strongly convex with some $M>0$, then
$F_{\mu}$ is ${M}/({1+\mu M})$-strongly convex. And further, the
minimization of $F(\vect{u})$ over $\vect{u} \in \mathcal{U}$ is
equivalent to that of $F_{\mu}(\vect{u})$ over $\vect{u} \in
\mathcal{U}$ in the sense that the set of minimizers of two problems
are the same.
\end{lemma}

From the definition of the smoothable function, we know that: 1) a
positive linear combination of smoothable functions is smoothable\footnote{If $F_1$ is smoothable with
  parameter $(a_1, b_1)$ and $F_2$ with parameter $(a_2,b_2)$, then
  $c_1 F_1 + c_2 F_2$ is smoothable with parameter $(c_1 a_1+c_2
  a_2, c_1 b_1 + c_2 b_2)$, for any $c_1, c_2 \ge 0$.}, and 2) the composition of a smoothable
function with a linear transformation is smoothable\footnote{Let
  $\map{A}{\mathcal{U}}{\mathcal{X}}$ be a linear transformation and
  let $\vect{b} \in \mathcal{X}$. Let $\map{\ell}{\mathcal{X}}{\real}$
  be a smoothable function with parameter $(a,b)$. Then, the
  function $\map{F}{\mathcal{U}}{\real}$, $ \vect{u} \mapsto \ell(A
  \vect{u} + \vect{b} )$ is smoothable with parameter $( a,
  b\Norm{A}^2)$, where $\Norm{A}:= {\max}_{\Norm{\vect{u}}=1} \Norm{A \vect{u}} $. If $\mathcal{X}=\real$, then $\Norm{A}$ is the
  $\ell_{\infty}$ norm. }. These properties enable us to smooth each
component of $G$, i.e., $\ell_1$, $\ell_2$, $h$ and $\Norm{\cdot}$,
which results in a smooth approximation of~\eqref{eq:P2} via the
corresponding $G_\mu$ as follows
\begin{equation}
  \begin{aligned}
      \min\limits_{\vect{u} \in \mathcal{U} } \; G_{\mu}(t,\vect{u} ).
  \end{aligned}
  \label{eq:P2smooth} \tag{P2$^\prime$}
\end{equation}
Note that $G_{\mu}$ is locally Lipschitz and minimizers
of~\eqref{eq:P2smooth} are that of~\eqref{eq:P2}. We provide in the
following lemma explicit expressions of~\eqref{eq:P2smooth} for the two problem classes.
\begin{lemma}[Examples of~\eqref{eq:P2smooth}] \label{lemma:smoothP} \\
  \noindent \textbf{Problem 1:} Consider the following loss function \begin{equation*}
  \ell(\vect{u}, \vect{x}):= \frac{1}{2}\Norm{\vect{u}}^2  +F_{\mu}(\vect{x}), \; \textrm{ given some } \mu >0,
\end{equation*}
where $\map{F_{\mu}}{\real^n}{\real}$ is  a smoothed $\ell_2$-norm function\footnote{
\textbf{The $\ell_2$-norm function:} \label{appx:func_norm}
Consider $ \vect{x} \in \real^n$,
    $\maps{F}{\vect{x}}{\Norm{\vect{x}}}$, and $\mu>0$. Clearly, $F$ is
    differentiable almost everywhere, except at the origin. Then,     \begin{equation*} \small
      \begin{aligned} \small
     & F_{\mu}(\vect{x}):=
     \min\limits_{\vect{z} \in \real^n} \left\{ \Norm{\vect{z}} +
       \frac{1}{2\mu}
       \Norm{\vect{z}- \vect{x} }^2   \right\}, \\
      &= \min\limits_{r \geq 0} \min\limits_{ \Norm{\vect{z}}=r }
      \left\{ r + \frac{1}{2\mu} \left( r^2 - 2 \trans{\vect{z}}
          \vect{x}
          + \Norm{\vect{x}}^2  \right)  \right\}, \\
      &= \min\limits_{r \geq 0} \left\{ r + \frac{1}{2\mu} \left( r^2 -
          2 r \Norm{\vect{x}} + \Norm{\vect{x}}^2  \right)
      \right\}, \\
      &= \begin{cases} \frac{\Norm{\vect{x}}^2}{2\mu}, &\textrm{ if }
        \; \Norm{\vect{x}} \leq \mu,  \\
        \Norm{\vect{x}}-\frac{\mu}{2} , &\textrm{ } \ow,
        \end{cases}
      \end{aligned}
    \end{equation*}
    with the smoothing parameter $({1}/{2}, 1)$.
}, with $\Lip(F_{\mu})=1$. Then,
the objective function $G_{\mu}(t,\vect{u} )$ is
\begin{equation*}
\begin{aligned}
  & \frac{1}{2} \Norm{\vect{u}}^2 + \frac{1}{T} \sum\limits_{k\in
    \mathcal{T}} F_{\mu} ( \vect{p}_{k,t} ) + \epsilon   + \frac{\gamma }{T} \sum\limits_{i=1}^{p} \sum\limits_{k\in
    \mathcal{T}} F_{\mu}( \vect{H}_k^{(i)} ),
  \end{aligned}
\end{equation*}
where $\vect{p}$, $\vect{H}$ are affine in $\vect{u}$, defined as in
Section~\ref{sec:tractable}. In addition, we have the smoothing parameter of $G_{\mu}(t,\vect{u} )$, $(a,b):=((1+ p \gamma)/2, \mu + s_0+
\gamma \sum_i s_i )$, where
\begin{equation*}
s_0= \sigma_{\max} \left( \trans{\left( \sum\limits_{i=1}^{p} \alpha_i f^{(i)}_2({t},\hat{\vect{x}}_{t}) \right)} \left( \sum\limits_{i=1}^{p} \alpha_i f^{(i)}_2({t},\hat{\vect{x}}_{t}) \right)  \right), \\
\end{equation*}
with $\sigma_{\max}$ denoting the maximum singular value of the
matrix, and
\begin{equation*}
s_i= \sigma_{\max} \left( \trans{ f^{(i)}_2({t},\hat{\vect{x}}_{t}) } f^{(i)}_2({t},\hat{\vect{x}}_{t})  \right), \; i\in \until{p}.
\end{equation*}
\noindent \textbf{Problem 2:}
Let us select the feature map $\phi$ to be
the identity map with the dimension $m=n$, and consider
\begin{equation*}
  \ell(\vect{u}, \vect{x}):= \max \{ 0, 1-  \langle \vect{u}, \vect{x}  \rangle   \} , \quad  \textrm{with } L(u)=  \Norm{\vect{u}}, \end{equation*}
resulting in
\begin{equation*}
  \begin{aligned}
   &G_{\mu}(t,\vect{u} )=  \frac{1}{T} \sum\limits_{k\in \mathcal{T}}    \supscr{F}{S}_{\mu}(\langle \vect{u}, \vect{p}_{k,t}  \rangle )
  +  q_{t}  F_{\mu}(\vect{u}),
  \end{aligned}
\end{equation*}
where $\mu>0$, parameters $\vect{p}$, $q$ are as in Section~\ref{sec:tractable}, and
functions $\map{\supscr{F}{S}_{\mu}}{\real}{\real}$ and $\map{F_{\mu}}{\real^n}{\real}$ are the smoothed switch function\footnote{
\textbf{The Switch function:} \label{appx:switch_func}
  Consider $ u \in \real$, $\maps{\supscr{F}{S}}{u}{\max \{ 0, 1- u \}
  }$, which is differentiable almost everywhere. For a given $\mu>0$,
  we compute
  \begin{equation*}
    \begin{aligned}
      & \supscr{F}{S}_{\mu}(u):= \min\limits_{z \in \real}
      \left\{ \max \{ 0, 1-  z  \} + \frac{1}{2\mu} \Norm{z- u }^2   \right\}, \\
      & = \min \left\{ \min\limits_{z \leq 1 } 1- z + \frac{1}{2\mu}
        \Norm{z- u }^2, \min\limits_{z \geq 1 }
        \frac{1}{2\mu} \Norm{z- u }^2   \right\} . \\
    \end{aligned}
  \end{equation*}
Given that
\begin{equation*}
  \begin{aligned}
    \min\limits_{z \leq 1 } 1- z + \frac{1}{2\mu} \Norm{z- u }^2
    =\begin{cases}
      \frac{1}{2\mu} \Norm{1- u }^2    ,  &\textrm{ if } \; u > 1- \mu,  \\
      1-u-\frac{\mu}{2} , &\textrm{ if } \; u \leq 1 -\mu,
  \end{cases}
  \end{aligned}
\end{equation*}
and
\begin{equation*}
  \begin{aligned}
    \min\limits_{z \geq 1 } \frac{1}{2\mu} \Norm{z- u }^2
    = \begin{cases}
      \frac{1}{2\mu} \Norm{1- u }^2    ,  &\textrm{ if } \; u<1,  \\
      0 , &\textrm{ if } \; u \geq 1,
\end{cases}
  \end{aligned}
\end{equation*}
resulting in
\begin{equation*}
  \begin{aligned}
    \supscr{F}{S}_{\mu}(u):= \begin{cases}
      1-u-\frac{\mu}{2} , &\textrm{ if } \;  u \leq 1 -\mu, \\
      \frac{1}{2\mu} \Norm{1- u }^2    ,  &\textrm{ if } \; 1-\mu \leq u<1,  \\
      0 , &\textrm{ if } \; u \geq 1,
\end{cases}
  \end{aligned}
\end{equation*}
  with the smoothing parameter $({1}/{2}, 1)$.
}
and $\ell_2$-norm function$\supscr{}{\ref{appx:func_norm}}$, respectively. Note that $G_{\mu}$ has the smoothing parameter $(a,b):=((1+ q_t)/2,
q_t + {1}/{T} \sum_{k\in \mathcal{T}} \Norm{\vect{p}_{k,t}
}^2_{\infty} )$. $\hfill$ \qed
\end{lemma}

\noindent \textbf{Step 2: (Solution to~\eqref{eq:P2smooth} as a
  dynamical system)} To solve~\eqref{eq:P2smooth} online, we propose a
dynamical system extending the Nesterov's accelerated-gradient
method by adapting gradients of the time-varying objective function.
In particular, let $\vect{u}_t$, $t\in \integernonnegative$, be solutions
of~\eqref{eq:P2smooth} and let us consider the solution system with
some $\vect{u}_{0} \in \mathcal{U}$ and $\vect{y}_{0}=\vect{u}_{0} $,
as
\begin{equation}
  \begin{aligned}
    \vect{u}_{t+1} =& \; \Pi_{\mathcal{U}}( \vect{y}_{t} -  \varepsilon_t \nabla G_{\mu}(t, \vect{y}_{t})  )   , \\
    \vect{y}_{t+1} =& \; \vect{u}_{t+1} + \eta_t  (\vect{u}_{t+1} -\vect{u}_{t}) , \\
  \end{aligned} \label{eq:oagsol}
\end{equation}
where $\varepsilon_t \leq {\mu}/{b_t}$ with positive parameters $\mu$
and $b_t:=b$ being those define $G_{\mu}(t,\vect{u})$. We denote by
$\nabla G_{\mu}$ the derivative of $G_{\mu}$ w.r.t. its second
argument and denote by $\Pi_{\mathcal{U}}(\vect{y})$ the projection of
$\vect{y}$ onto $ \mathcal{U}$ as in Assumption~\ref{assump:cvxu} on convex decision oracle. Note that, the gradient function $\nabla G_{\mu}$ can be computed in closed form for problems of interest, see, e.g., Appendix~\ref{appx:gradients} for those of the proposed problems. Further, we
select the moment coefficient $\eta_t \in\realnonnegative$ as in
Appendix~\ref{appx:stability}.  In the following, we leverage
Appendix~\ref{appx:stability} on the stability analysis of the
solution system~\eqref{eq:oagsol} for a regret bound between online
decisions and optimal solutions of~\eqref{eq:P1}. \begin{theorem}[Probabilistic regret bound
  of~\eqref{eq:P1}] \label{thm:regret} Given any $t \geq 2$, let us
  denote by $\vect{u}_{t}$ and $\vect{u}_{t}^{\star}$ the decision
  generated by~\eqref{eq:oagsol} and an optimal solution which solves
  the online Problem~\eqref{eq:P1}, respectively. Consider the dynamic
  regret to be the difference of the cost expected to incur if we
  implement $\vect{u}_{t}$ instead of $\vect{u}_{t}^{\star}$, defined
  as \begin{equation*}
  R_t:= \mathbb{E}_{\probP_{t+1|t}} \left[ \ell(\vect{u}_{t}, \vect{x})\right] -
  \mathbb{E}_{\probP_{t+1|t}} \left[ \ell(\vect{u}_{t}^{\star}, \vect{x})\right].
\end{equation*}
Then, the regret $R_t$ is bounded in probability as
follows \begin{equation*}
  \begin{aligned}
    & \Prob \left( R_t \leq \frac{4W_{t}}{(t+2)^2} + T F_{t} + a \mu +
     2 L({\vect{u}}_{t}^{\star}) \hat{\epsilon} \right) \geq \rho(t),
  \end{aligned}
\end{equation*}
where $W_{t}$ depends on the system state at time
$t-T$,           and $F_{t}$ depends on the variation of the optimal objective values in $\mathcal{T}$, i.e.,
\begin{equation*}
  \begin{aligned}
    F_{t}=& \max_{k\in \mathcal{T}} \left\{ | G_{k+1}^{\star} -
      G_{k}^{\star} | \right\} + \bar{L} ,
  \end{aligned}
\end{equation*}
where $G_k^{\star}:= G(k,{\vect{u}}_{k}^{\star} )$ is the optimal
objective value of~\eqref{eq:P2}, or equivalently that
of~\eqref{eq:P1}. Further, $\bar{L}$ is the variation bound of $G$
w.r.t. time, and
the rest of the parameters are the same as before.
Furthermore, if all historical data are assimilated for the decision $\vect{u}_{t}$, then, we have
\begin{equation*}
  \begin{aligned}
    & \liminf\limits_{t \rightarrow \infty} \Prob \left( R_t \leq T F_{t} + a \mu \right) \geq 1-\beta,
  \end{aligned}
\end{equation*}
with $\beta$ a given, arbitrary confidence value.
\end{theorem}
Theorem~\ref{thm:regret} quantifies the dynamic regret of online
decisions $\vect{u}$ w.r.t. solutions to~\eqref{eq:P1} in high
probability. Notice that, the regret bound is dominated by terms: $T
F_t$, $a \mu$ and $L({\vect{u}}_{t}^{\star})\hat{\epsilon}$, which
mainly depend on three factors: the data-driven parameters
$\varepsilon$, $\eta$ and $\mu$ of the solution
system~\eqref{eq:oagsol}, the variation $F_t$ over optimal objective values, and the parameters $T$, $\beta$,
$\gamma$ and $\hat{\epsilon}$ that are related to the system and environment
learning. In practice, a small regret bound is determined by 1) an
effective learning procedure which contributes to small
$\hat{\epsilon}$; 2) a proper selection of the loss function $\ell$
which results in smoothing procedure with a small parameter $a\mu$;
and 3) the problem structure leading to small variations $F_t$ of the optimal objectives values. Furthermore, when we use all the historical data for the objective gradients in the solution system~\eqref{eq:oagsol}, the effect of system ambiguity learning is negligible asymptotically.

\noindent \textbf{Online Procedure:}
Our online algorithm is summarized in the Algorithm~\ref{alg:Optal}.

{\begin{algorithm}[hpt]
\floatname{algorithm}{Online Optimization and Learning Algorithm}{}
    \caption{${\textrm{Opal}}(\mathcal{I} )$}
\label{alg:Optal}
\begin{algorithmic}[1]
  \State Select $\{ f^{(i)}\}_i$, $\ell$, $\beta$, $\mathcal{U}$,
  $\vect{u}_0$, $\mu$, and $t = 1$;
\Repeat \State Update data set $\mathcal{I}:=\mathcal{I}_t$; \State
Compute $\vect{\alpha}:=\vect{\alpha}_t$ as in~\cite{DL-DF-SM:20-lcss}; \State Select $\hat{\probP}_{t+1|t}$
in~\eqref{eq:empdistn} and
$\hat{\epsilon}:=\hat{\epsilon}(t,T,\beta,\vect{\alpha},\vect{u})$ in~\eqref{eq:epsilon};
\State Run dynamical system~\eqref{eq:oagsol} for
$\vect{u}:=\vect{u}_t$;
\State Apply $\vect{u}$ to~\eqref{eq:P} with the regret
guarantee; \State $t \leftarrow t+1$;
\Until time $t$ stops.
\end{algorithmic}
\end{algorithm}}
\section{Implementation}\label{sec:simulation}
In this section, we apply our algorithm to the introduced motivating examples, resulting in online-tractable, effective system learning with guaranteed, regret-bounded performance in high probability.
\subsection{Optimal control of an uncertain nonlinear system}
We consider the two-wheel vehicle driving under various road conditions, and our goal is to learn one-step prediction of the system state distribution and leverage for path tracking under various unknown road zones. In particular, we represent the two-wheel vehicle as a differential-drive robot subject to uncertainty~\cite{SML:06}:
\begin{equation}
  \begin{aligned}
  {x}_{t+1}=& {x}_{t} + h \cos({\theta}_{t}) {d}_{1,t} + h{w}_{1,t} , \\
  {y}_{t+1}=& {y}_{t} + h \sin({\theta}_{t}) {d}_{1,t} + h {w}_{2,t}, \\
  {\theta}_{t+1}=& {\theta}_{t} - h {d}_{2,t} + h {w}_{3,t}, \\
  {d}_{1,t}=& \frac{r}{2} ( v_{l,t} + v_{r,t} + e_{1,t}   ), \\
  {d}_{2,t}=& \frac{r}{2R} ( v_{l,t} - v_{r,t} + e_{2,t}   ),
\end{aligned} \label{eq:unicycle}
\end{equation}
where components of states $\vect{x}_t:=(x_t, y_t, \theta_t) \in \real^2 \times [-\pi,\pi) \cong
\real \times \unitcircle $ represent vehicle position and orientation on the 2-D plane. We take the discretization parameter $h=0.01$ and assume subGaussian
uncertainty $\vect{w}_t:=(w_{1,t},w_{2,t},w_{3,t}) \in \real^3$ to be a zero-mean, mixture of
Gaussian and Uniform distributions with $\sigma=0.5$.
The intermediate variable
$\vect{d}_t:=(d_{1,t},d_{2,t})$ depends on the wheel radius $r=0.15$ m, the
distance between wheels $R=0.4$ m, the controlled left-right wheel speed
$\vect{u}_t:=(v_{l,t}$, $v_{r,t})$ and an unknown parameter $
\vect{e}_t:=(e_{1,t} , e_{2,t})$, which depends on the wheel quality and road conditions. For simplicity, we assume that the planner adapts the system~\eqref{eq:unicycle} with $\vect{e}_t\equiv(0,0)$ and $\vect{w}_t\equiv(0,0,0)$, and the vehicle can move over three types of road zones, the regular zone with $\vect{e}^{(1)}:=(0,0)$, the slippery zone with $\vect{e}^{(2)}=(4,0)$, and the sandy zone with
$\vect{e}^{(3)}=(-1.2,-0.2)$, where locations of these zones are described in
Fig.~\ref{fig:path_plan}.

To adapt the proposed approach, we consider Problem~\eqref{eq:P} with the following loss function
\begin{equation*}
  \begin{aligned}
\ell(\vect{u},x,y,\theta)=& \frac{1}{20}\Norm{\vect{u}-\supscr{\vect{u}}{ref} }^2 + \frac{1}{14 \sqrt{2}} |x-\supscr{x}{ref}| + \\
& \frac{1}{4 \sqrt{2}} |y-\supscr{y}{ref}| + \frac{289}{8} \left( \cos(\theta) - \cos(\supscr{\theta}{ref}) \right)^2 + \\
& \frac{289}{8} \left( \sin(\theta) - \sin(\supscr{\theta}{ref}) \right)^2,
  \end{aligned}
\end{equation*}
where $(\supscr{\vect{u}}{ref},\supscr{x}{ref},\supscr{y}{ref},\supscr{\theta}{ref})$ are signals generated by the planner, and we select the parameter $\mu=10^{-4}$ for components which are not smooth. In addition, we assume $\mathcal{U}=[-20,20]^2$ and utilize $p=3$  basis functions $\{f^{(i)} \}_i$ in form of~\eqref{eq:unicycle}, with
$\vect{w}_t\equiv (0,0,0)$, and
\begin{equation*}
  \begin{aligned}
  &  i=1, \quad e_1=0, \; & e_2=0, \quad \\
  &  i=2, \quad e_1=10, \; &  e_2=0, \quad \\
  &  i=3, \quad e_1=0, \; &  e_2=10. \;\;
  \end{aligned}
\end{equation*}
Note that the ground truth parameter $\vect{\alpha}^{\star}:=(1,0,0)$ in the regular zone, $\vect{\alpha}^{\star}:=(0.6,0.4,0)$ in the slippery zone,
and $\vect{\alpha}^{\star}:=(1.14,-0.12,-0.02)$ in the sandy zone.
At each time $t$, we have access to model sets
$\{f^{(i)} \}_i$ and as well as the real-time data set $\mathcal{I}_t$ with size $T_0=100$, which corresponds to the moving time window of order 0.1 second. For the system learning algorithm, notions of norm and inner product are those defined on the vector space $T(\real^2 \times \mathbb{S}) \equiv \real^3$.
We employ our online optimization and learning algorithm for the
characterization of the uncertain vehicle states, learning of the unknown road-condition parameter $\vect{e}$, and control towards planned behaviors in real time. The achieved system behaviors are demonstrated in Fig~\ref{fig:path_achieved}, contrasted with the case without the proposed approach, as in Fig.~\ref{fig:path_plan}. In the following, we analyze each case separately and notice how the proposed approach strikes balance between the given planned control $\supscr{\vect{u}}{ref}$ and the actual control $\vect{u}$ which reduces the weighted tracking error in road uncertainty.

\begin{figure}[tbp]\centering
\includegraphics[width=0.125\textwidth]{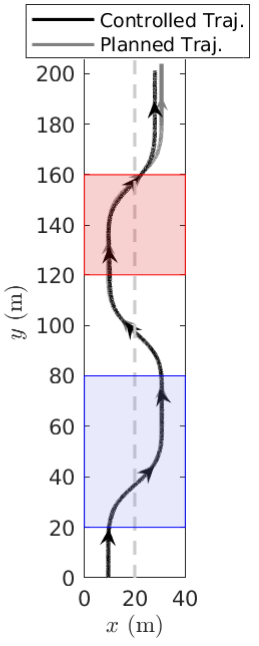}\includegraphics[width=0.405\textwidth]{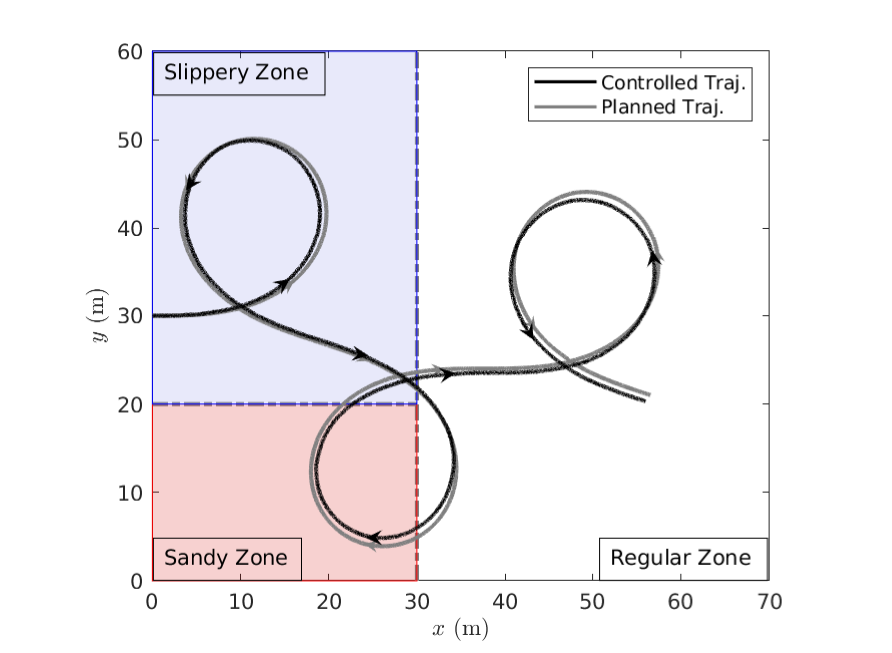}\caption{\footnotesize An example of the (gray) planned trajectory and (black) controlled system trajectory in various road zones, with the system state $\vect{x}=(x,y,\theta)$. The red region indicates sandy zone while the blue region indicates the slippery zone. With the implemented control, the vehicle follows the planned path with low regrets in high probability.
}\label{fig:path_achieved}\end{figure}

\textbf{Example (Lane-changing behavior adaptation)} In this scenario,
we assume the initial system state $\vect{x}_0=(10,0,\pi/2)$. Further, the
vehicle can access path plan in Fig.~\ref{fig:path_plan}(a) and as
well as the suggested wheel speed plan as the gray signal in
Fig.~\ref{fig:control_lane}(a). To demonstrate the learning effect of the
algorithm, we show in Fig.~\ref{fig:car_alpha} components $\alpha_1$
and $\alpha_2$ of $\vect{\alpha}=(\alpha_1,\alpha_2,\alpha_3)$, where
the black lines indicate value of the ground truth
$\vect{\alpha}^{\star}$ on the planned trajectory and the gray lines
represent the learned, real-time estimate of $\alpha_1$ and $\alpha_2$
at the actual vehicle position. Notice that $\vect{\alpha}^{\star}$ is
inaccessible in practice, and from this case study, the proposed
approach indeed learns the system dynamics effectively. See,
e.g.~\cite{DL-DF-SM:20-lcss} for more analysis regarding to the effect
of the learning behavior and ambiguity sets characterization on the
selection of $\epsilon$ and $\gamma$.

As the proposed loss function $\ell$ measures the weighted tracking
error, the resulting control system trajectory in
Fig.~\ref{fig:path_achieved}(a) already reveals the effectiveness of
the method and as well as the low regrets in probability. On the other
hand, because the system is highly non-linear and uncertain,
evaluating the actual optimal objective value of Problem~\eqref{eq:P}
is difficult. Therefore, it's very challenging to evaluate the regret
$R_t$ in practice, even though the its probabilistic bounded is
proved. Here, we provide in Fig.~\ref{fig:control_lane}(b)  the realized loss
$\ell$ and as well as the realized objective value of
Problem~\eqref{eq:P2}, where the loss $\ell$ reveals one possible
objective value of~\eqref{eq:P}, and the objective value
of~\eqref{eq:P2} serves as an upper-bound of that of~\eqref{eq:P} in
high probability.       In addition, notice that the derived (black) control signal in
Fig.~\ref{fig:control_lane}(a) has undesirable, high-oscillatory
behavior. This is because the chosen loss function $\ell$ is only
locally convex in $\vect{x}$. When the system disturbances are
significant, the proposed approach then revealed certain degradation
and control being oscillatory. Nevertheless, a desirable system
behavior in Fig.~\ref{fig:path_achieved}(a) is achieved.

\begin{figure}[tbp]\centering
\includegraphics[width=0.24\textwidth]{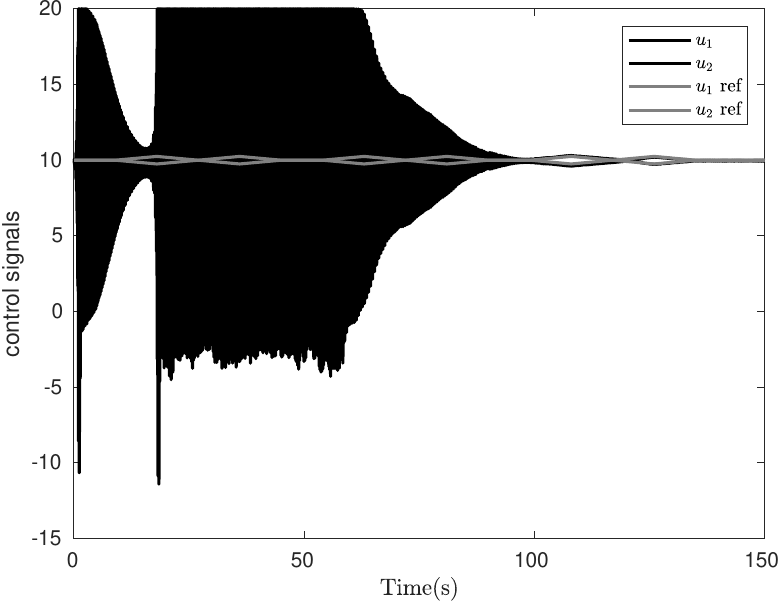}\includegraphics[width=0.24\textwidth]{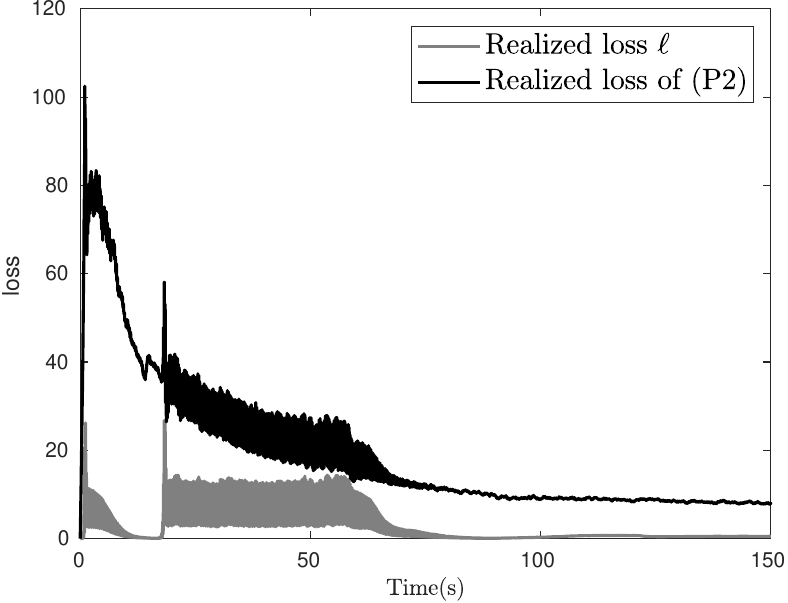}\\ $\quad$ (a)  $\qquad \qquad \qquad\qquad\qquad\quad$  (b)
\caption{\footnotesize (a) The (gray) control signal provided by the planner and an example of the (black) control signal derived from the proposed approach. (b) The realized loss $\ell$ and the achieved objective of~\eqref{eq:P2}.
}\label{fig:control_lane}\end{figure}

\begin{figure}[tbp]\centering
\includegraphics[width=0.235\textwidth]{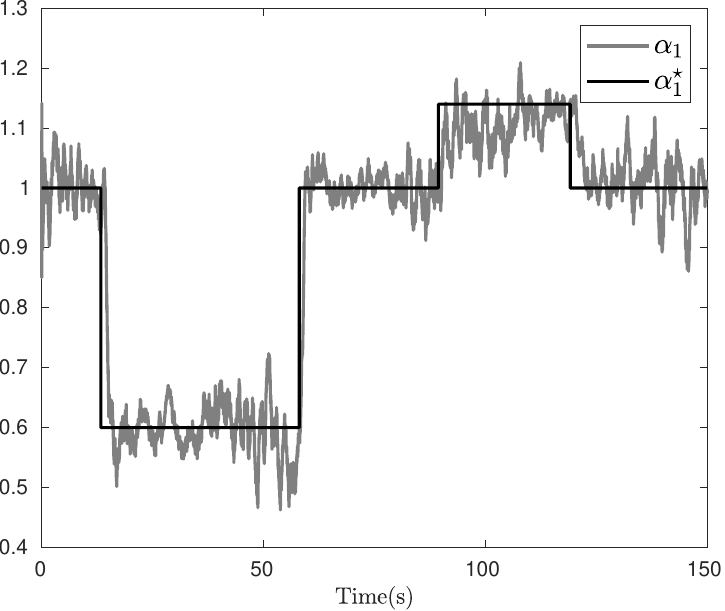} \includegraphics[width=0.245\textwidth]{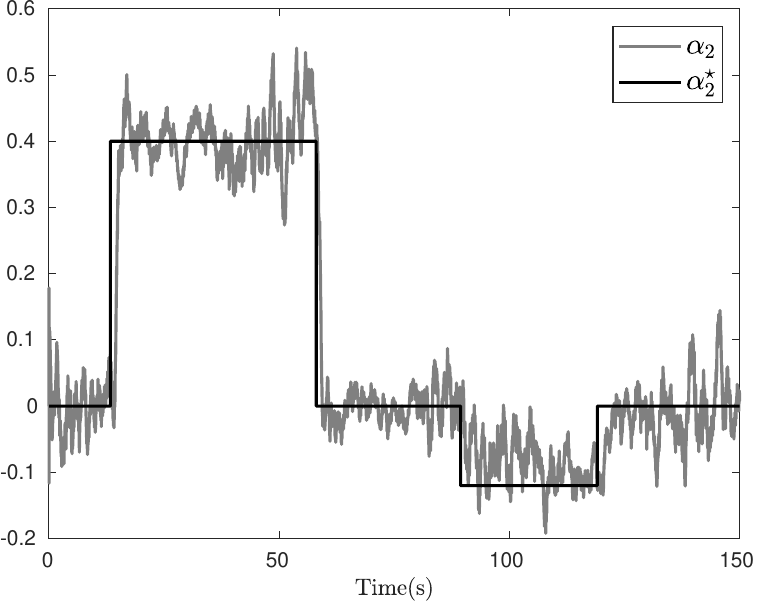}\caption{\footnotesize The component $\alpha_1$ and $\alpha_2$ of the real-time parameter $\vect{\alpha}:=(\alpha_1,\alpha_2,\alpha_3)$ in the learning procedure.}\label{fig:car_alpha}\end{figure}

\textbf{Example (Circular route tracking)}
In this scenario, we consider $\vect{x}_0=(0,30,0)$. We omit the details as the analysis shares the same spirit as the last lane-changing example.

\subsection{Online resource allocation
  problem} \label{subsec:resource} We consider an online resource
allocation problem where an agent or decision maker aims to 1) achieve
at least target profit under uncertainty, and 2) allocate resources as
uniformly as possible. To do this, the agent distributes available
resources, e.g., wealth, time, energy or human resources, to various
projects or assets. In particular, for the trading-market motivating example, let us consider that the agent
tries to make an online allocation $\vect{u}\in \mathcal{U}$ of a unit
wealth to three assets. At each time $t$, the agent receives random
return rates $\vect{x}_t \in \realnonnegative^3$ of assets from some unknown and uncertain dynamics
\begin{equation}
  \vect{x}_{t+1}  =
\vect{x}_t +  hA(t) + h \vect{w}_t, \textrm{ with some } \vect{x}_{0} \in \real^3,
 \label{eq:allocation}
\end{equation}
where $h=10^{-3}$ is a stepsize, the vector $A(t)$ is randomly
generated, unknown and piecewise constant, and the uncertainty vector
$\vect{w}_t$ is assumed to be sub-Gaussian with $\sigma=0.1$. Note that
this model can serve to characterize a wide class of dynamic (linear and
nonlinear) systems. In addition, we assume that the third asset is
value preserved, i.e., the third component of $A(t)$ and $\vect{w}_t$
are zero and $x_3 \equiv 1$. Over time, an example of the resulting unit return
rates $\vect{x}$ is demonstrated in Fig.~\ref{fig:assets}. Then, we
denote by $r_0=1.3$ and $\langle \vect{u}, \vect{x}_{t+1} \rangle$ the
target profit and the predicted instantaneous profit, respectively. Note that
the decision maker aims to obtain at least a $30 \%$ profit and
allocate resources online for this purpose. In particular, the
decision maker implements an allocation online if $\langle \vect{u},
\vect{x}_{t+1} \rangle \leq r_0$, otherwise does nothing. This results
in~\eqref{eq:P} with the loss function
\begin{equation*}
\ell(\vect{u},\vect{x})= \max \{ 0, 1-  \frac{1}{r_0}\langle \vect{u}, \vect{x}  \rangle   \},
\end{equation*}
and set $\mathcal{U}$ a unit simplex. We propose $p=3$  basis functions
\begin{equation*}
  f^{(1)}=\vect{x}, \; f^{(2)}=\vect{x}+ 0.1 h \vect{e}_1 ,\;  f^{(3)}=\vect{x}+ 0.1 h \vect{e}_2,
        \end{equation*}
where $\vect{e}_1=\trans{(1,0,0)}$ and $\vect{e}_2=\trans{(0,1,0)}$.
At each $t$, we assume that only historical data are available for online resource allocations. Applying the proposed
probabilistic characterization of $\vect{x}_{t+1}$ as
in~\eqref{eq:P1}, we equivalently write it as in
form~\eqref{eq:P2smooth}, where
\begin{equation*}
  \begin{aligned}
    &G_{\mu}(t,\vect{u} )= \frac{1}{T} \sum\limits_{k\in
      \mathcal{T}}     \supscr{F}{S}_{\mu}(\langle \vect{u}, \frac{\vect{p}_{k,t}}{r_0}
    \rangle ) + \frac{q_{t}}{r_0} F_{\mu}(\vect{u}), \; \mu=0.01,
  \end{aligned}
\end{equation*}
with functions $\supscr{F}{S}_{\mu}$$\supscr{}{\ref{appx:switch_func}}$ and $F_{\mu}$$\supscr{}{\ref{appx:func_norm}}$, and
 real-time data $\vect{p}_{k,t}$ and $q_{t}$ determined as in
Problem 2. We claim that $G_{\mu}(t,\vect{u} )$ has a time-dependent
Lipschitz gradient constant in $\vect{u}$ given by $\Lip(G_{\mu})=q_t
/ {r_0} + {1}/{(r_0^2 T)} \sum_{k\in \mathcal{T}} \Norm{\vect{p}_{k,t}
}^2_{\infty}$, and we use $\varepsilon:=1/\Lip(G_{\mu})$ in the
solution system~\eqref{eq:oagsol} to compute the online decisions.
\begin{figure}[tbp]\centering
\includegraphics[width=0.245\textwidth]{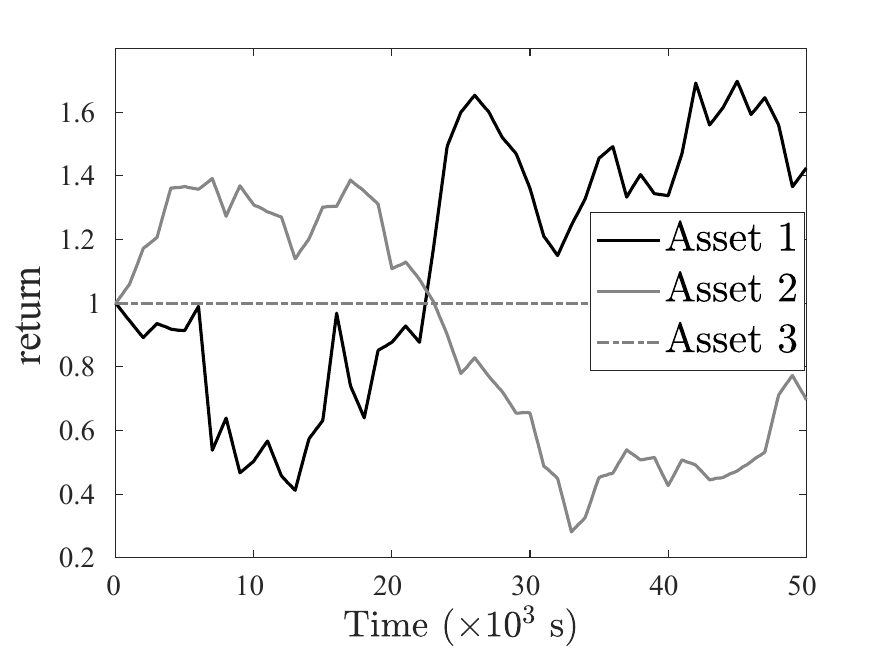}\includegraphics[width=0.225\textwidth]{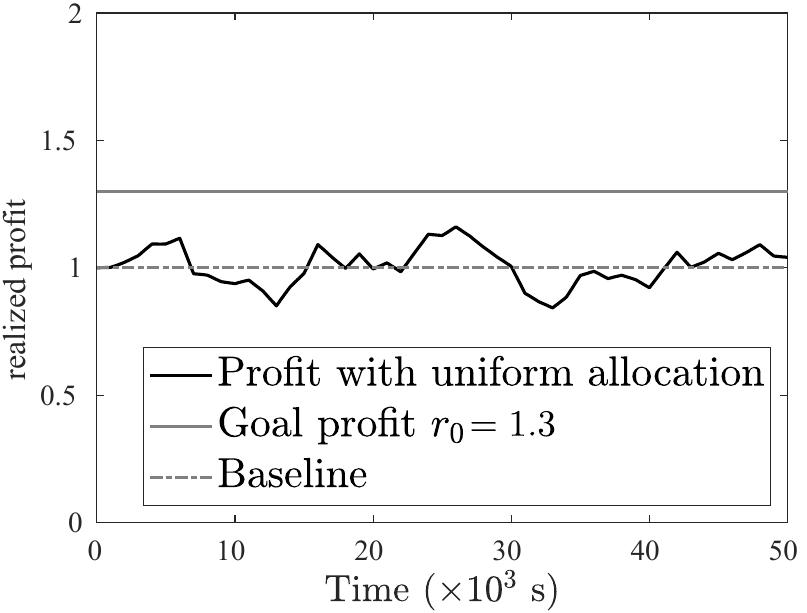}\caption{\footnotesize An example of random returns
 $\vect{x}=(x_1,x_2,x_3)$, where returns of the first two assets
 $x_1$, $x_2 \in [0, +\infty)$ are highly fluctuating and the third is
 value-preserving with return $x_3\equiv 1$. Without asset allocation, agent does not achieve the goal profit $r_0=1.3$ and has a chance of losing assets.}\label{fig:assets}\end{figure}
\begin{figure}[tbp]\centering
\includegraphics[width=0.22\textwidth]{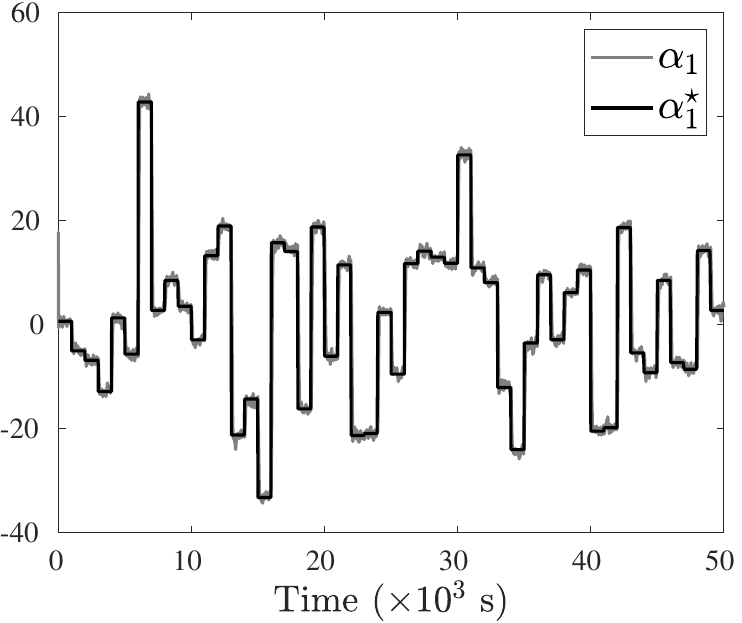} \includegraphics[width=0.23\textwidth]{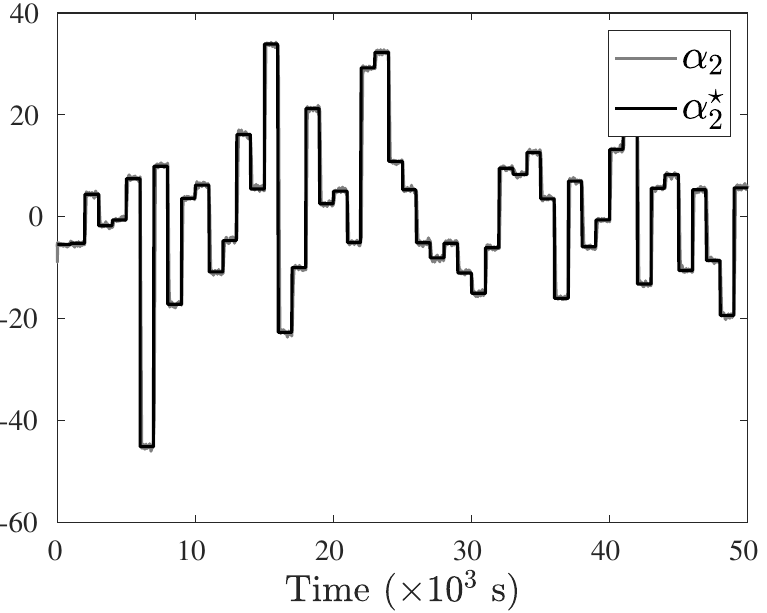}\caption{\footnotesize The component $\alpha_1$ and $\alpha_2$ of the real-time parameter $\vect{\alpha}:=(\alpha_1,\alpha_2,\alpha_3)$ in learning, where the values $\alpha^{\star}_1$ and $\alpha^{\star}_2$ are the online-inaccessible ground truth. Notice the responsive behavior of the proposed learning algorithm.
}\label{fig:alpha_allocation}\end{figure}
\begin{figure}[tbp]\centering
\includegraphics[width=0.23\textwidth]{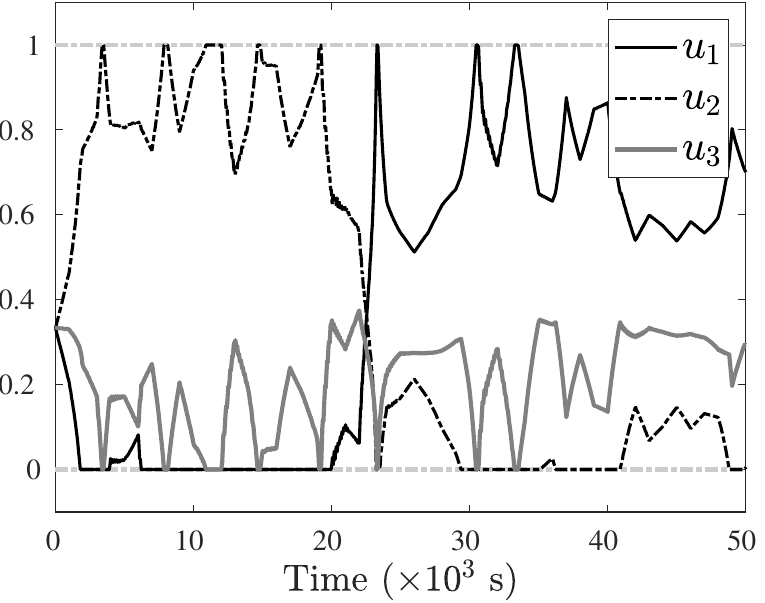}\includegraphics[width=0.23\textwidth]{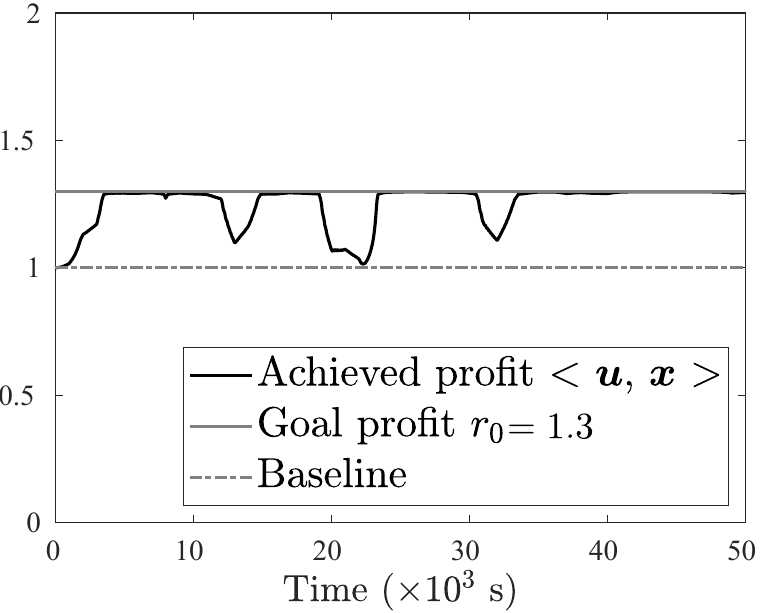}\caption{\small Real-time resource allocation $\vect{u}$ and profit $\langle \vect{u}, \vect{x}  \rangle$. Notice how the decision $\vect{u}=(u_1,u_2,u_3)$ respects constraints and how the allocation tries to balance the assets when the goal profit $r_0$ is met.
}\label{fig:control_allocation}\end{figure}

Fig.~\ref{fig:alpha_allocation} shows the real-time evolution $\alpha_1$ and $\alpha_2$ of the
parameter $\vect{\alpha}:=(\alpha_1,\alpha_2,\alpha_3)$, while the
behavior of $\alpha_3$ can be similarly characterized. In these
figures, black lines $\alpha^{\star}_1$ and $\alpha^{\star}_2$ are
determined by the unknown signal $A(t)$ while gray lines $\alpha_1$ and $\alpha_2$ are those
 computed as in~\cite{DL-DF-SM:20-lcss}. Note that
$\vect{\alpha^{\star}}$ represents the unknown
dynamics $f$ and they are not accessible in reality. It can be seen
that the proposed method effectively learns $\vect{\alpha}^{\star}$.

Fig.~\ref{fig:control_allocation} demonstrates the online resource
allocation obtained by implementing~\eqref{eq:oagsol} and the achieved
real-time profit $\langle \vect{u}, \vect{x} \rangle$. The decision
$\vect{u}$ starts from the uniform allocation
$\vect{u}_{0}=(1/3,1/3,1/3)$ and is then adjusted to approach the
target profit $r_0=1.3$.  Once the target is achieved, the agent then
maintains the profit while trying to balance the allocation if
possible. When the return rate $\vect{x}$ is low/unbalanced,
as in Fig.~\ref{fig:assets}, the agent tries to improve
and achieve the target profit by allocating resources more
aggressively. Though did not appear in the current
  scenario, in case that the return rate is high and the target
profit value is achieved, the agent focuses on balancing the
allocation while maintaining the profit. If both the target
  profit and allocation balance are achieved, then the agent stops
re-allocating resources and monitors the return rate $\vect{x}$ until
the switch turns on, e.g., when the near future profit prediction
drops below $r_0$ again.  In addition, notice how the target profit was achieved with
  the proposed control strategy as demonstrated in
  Fig.~\ref{fig:control_allocation}, which contrasts with uniform
  allocation case as in Fig.~\ref{fig:assets}.
\begin{figure}[tbp]\centering
\includegraphics[width=0.26\textwidth]{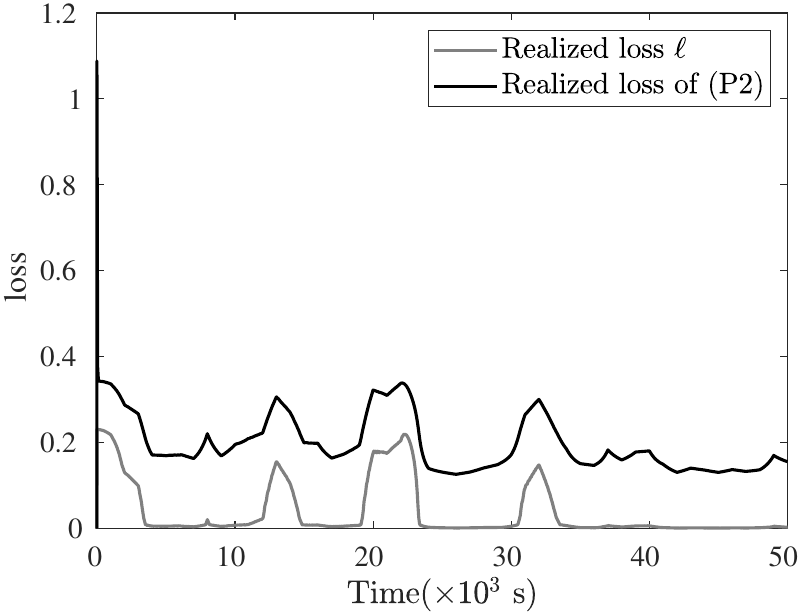} \caption{\footnotesize The realized loss $\ell$ and the achieved objective of~\eqref{eq:P2}. }\label{fig:asset_loss}\end{figure}

Fig.~\ref{fig:asset_loss} demonstrates the evaluation
  of the time-varying loss $\ell$ as well as the realized objective
  value of Problem~\eqref{eq:P2}. Due to the unknown time-varying
  distributions $\probP_{t|t-1}$, the evaluation of the objective values of
  Problem~\eqref{eq:P} is intractable, and the realized loss
  of~\eqref{eq:P2} serves as a high-confidence upper bound of that
  of\eqref{eq:P}. Nevertheless, the target profit is achieved with low
  regret in high confidence, as revealed  in
  Fig.~\ref{fig:control_allocation}.
  \vspace*{-2ex}
\section{Conclusions}\label{sec:conclusion}
In this paper, we proposed a unified solution framework for online
learning and optimization problems in form of~\eqref{eq:P}. The
proposed method allowed us to learn an unknown and uncertain dynamic
system, while providing a characterization of the system
with online-quantifiable probabilistic guarantees that certify the
performance of online decisions. The approach provided tractable,
online convex version of~\eqref{eq:P}, via a series of equivalent
reformulation techniques. We explicitly demonstrated the framework via
two problem classes conforming to~\eqref{eq:P}: an optimal control
problem under uncertainty and an online resource allocation
problem. These two problem classes resulted in explicit, online and
non-smooth convex optimization problems. We extended Nesterov's
accelerated-gradient method to an online fashion and provided a
solution system for online decision generation of~\eqref{eq:P}.  The
quality of the online decisions were analytically certified via a
probabilistic regret bound, which revealed its relation to the
learning parameters and ambiguity sets. Two motivating examples applying the proposed framework were empirically tested, demonstrating the effectiveness of the proposed framework with the bounded regret guarantees in probability.
We leave the relaxation of assumptions and the comparison of this work with other methods as the future work.
 \vspace*{-2ex}
\appendix \label{appx:whole}
\vspace*{-1ex}
\subsection{Computation of the objective
  gradients} \label{appx:gradients}
  \vspace*{-1ex}
Let $\ell$, $G$ and $G_{\mu}$ be those in Lemma~\ref{lemma:smoothP} on examples of~\eqref{eq:P2smooth}. We now derive $\nabla G_{\mu}:=\nabla_{\vect{u}} G_{\mu}(t,\vect{u} )$ as follows. \\
\noindent \textbf{Problem 1 (Optimal control under uncertainty): }
  \begin{equation*}
    \begin{aligned}
    & \nabla_{\vect{u}} G_{\mu}(t,\vect{u} )= \\
      & \hspace{4ex} \frac{1}{\mu} \vect{u} + \frac{1}{T} \sum\limits_{k\in
        \mathcal{T}} \nabla_{\vect{u}} F_{\mu} ( \vect{p}_{k,t} ) +
      \frac{\gamma }{T} \sum\limits_{i=1}^{p} \sum\limits_{k\in
        \mathcal{T}} \nabla_{\vect{u}} F_{\mu}( \vect{H}_k^{(i)} ),
    \end{aligned}
  \end{equation*}
  where, for each $k \in \mathcal{T}$, the term $\nabla_{\vect{u}}
  F_{\mu} ( \vect{p}_{k,t})$ is
  \begin{equation*}
    \begin{aligned}
            \begin{cases}
        \frac{1}{\mu} \trans{\left( \sum\limits_{i=1}^{p} \alpha_i f^{(i)}_2({t},\hat{\vect{x}}_{t}) \right) } \vect{p}_{k,t},  &\textrm{if } \; \Norm{\vect{p}_{k,t}} \leq \mu,  \\
        \frac{1}{\Norm{\vect{p}_{k,t}} } \trans{\left(
            \sum\limits_{i=1}^{p} \alpha_i
            f^{(i)}_2({t},\hat{\vect{x}}_{t}) \right)} \vect{p}_{k,t} ,
        &\; \ow,
      \end{cases} \\
    \end{aligned}
  \end{equation*}
  and, for $k \in \mathcal{T}$, $i\in \until{p}$, the term $ \nabla_{\vect{u}} F_{\mu} (  \vect{H}_k^{(i)} )$ is
  \begin{equation*}
    \begin{aligned}
            \begin{cases}
        - \frac{1}{\mu} \trans{(f^{(i)}_2(t,\hat{\vect{x}}_{t}) )} \vect{H}_k^{(i)},  &\textrm{if } \; \Norm{\vect{H}_k^{(i)}} \leq \mu,  \\
        - \frac{1}{\Norm{\vect{H}_k^{(i)}} } \trans{(f^{(i)}_2(t,\hat{\vect{x}}_{t}) )} \vect{H}_k^{(i)}  , &\; \ow.
      \end{cases}
    \end{aligned}
  \end{equation*}
  \noindent \textbf{Problem 2 (Online resource allocation): }
\begin{equation*}
  \begin{aligned}
    & \nabla_{\vect{u}} G_{\mu}(t,\vect{u} )= \frac{1}{T}
    \sum\limits_{k\in
      \mathcal{T}}     \nabla_{\vect{u}} \supscr{F}{S}_{\mu}(\langle \vect{u},
    \vect{p}_{k,t} \rangle ) + q_{t} \nabla_{\vect{u}}
    F_{\mu}(\vect{u}),
  \end{aligned}
\end{equation*}
where
    \begin{equation*}
      \nabla_{\vect{u}} F_{\mu} ( \vect{u} ):= \begin{cases}
        \frac{1}{\mu}  \vect{u},  &\textrm{if } \; \Norm{\vect{u}} \leq \mu,  \\
        \frac{1}{\Norm{\vect{u}} } \vect{u}  , &\; \ow,
        \end{cases}
    \end{equation*}
    and, for each $k \in \mathcal{T}$, the gradient $\nabla_{\vect{u}}
    \supscr{F}{S}_{\mu}(\langle \vect{u}, \vect{p}_{k,t} \rangle )$ is
    \begin{equation*}
      \begin{aligned}
                  \begin{cases}
         - \vect{p}_{k,t}   , &\textrm{ if } \;  \langle \vect{u}, \vect{p}_{k,t}  \rangle \leq 1 -\mu, \\
    - \frac{1-\langle \vect{u}, \vect{p}_{k,t} \rangle
    }{\mu}     \vect{p}_{k,t},  &\textrm{ if } \; 1-\mu \leq \langle \vect{u}, \vect{p}_{k,t}  \rangle <1,  \\
    0 , &\textrm{ if } \; \langle \vect{u}, \vect{p}_{k,t} \rangle
    \geq 1.
  \end{cases}
\end{aligned}
\end{equation*}
These explicit expressions provide ingredients for the solution system. With different
selections of the norm, the expression varies accordingly.

\vspace*{-2ex}
\subsection{Stability Analysis of the Solution
  System} \label{appx:stability}
\vspace*{-1ex}
Here, we adapt dissipativity theory
to address the performance of the online solution
system~\eqref{eq:oagsol}. This part of the work is an
online-algorithmic extension of the existing Nesterov's accelerated-gradient
method and its convergence analysis
in~\cite{BH-LL:17,LL-BR-AP:16,AB-MT:09}. Our extension~\eqref{eq:oagsol} inherits from the work in~\cite{LL-BR-AP:16}, where the difference is that gradient computations in~\eqref{eq:oagsol} are from time-varying objective functions in~\eqref{eq:P2smooth}.
  To simplify the discussion, the notation we used in this
subsection is different from that in the main body of the
paper. Consider the online problem, analogous to~\eqref{eq:P2smooth},
defined as follows
\begin{equation}
  \begin{aligned}
      \min\limits_{\vect{x} \in \mathcal{X} } \; f_{t}(\vect{x}), \quad t=0,1, 2,\ldots
  \end{aligned}
  \label{eq:P2sim} \end{equation}
where $f_{t}(\vect{x})$ is locally Lipschitz in $t$ with the parameter
$h(\vect{x})$ and, at each time $t$, the objective function $f_t$ are
$m_t$-strongly convex and $L_t$-smooth, with $m_t \geq 0$ and
$L_t>0$. The convex set $\mathcal{X} \subset \real^n$ is analogous to
that in Assumption~\ref{assump:cvxu} on convex decision oracle. The solution system
to~\eqref{eq:P2sim}, analogous to~\eqref{eq:oagsol}, is \begin{equation}
  \begin{aligned}
    \vect{x}_{t+1} =& \; \Pi( \vect{y}_{t} - \alpha_t \nabla f_{t}( \vect{y}_{t})  )   , \\
    \vect{y}_{t+1} =& \; \vect{x}_{t+1} + \beta_{t+1} \; (\vect{x}_{t+1} -\vect{x}_{t}) , \\
    & \textrm{with some } \vect{y}_{0}= \vect{x}_{0} \in \mathcal{X},
  \end{aligned} \label{eq:solsim}
\end{equation}
where $\alpha_t \leq 1/ L_t$ and $\beta_t$ is selected iteratively,
following
\begin{equation*}
  \delta_{-1}=1, \; \delta_{t+1}:=\frac{1+ \sqrt{1+4\delta_t^2}}{2},
  \; \beta_t:=\frac{\delta_{t-1}-1}{\delta_t}.
\end{equation*}
Note that $\delta_{t}^2 -\delta_{t}=\delta_{t-1}^2$,
$t=0,1,2,\ldots$. The projection $\Pi(\vect{x} )$ at each time
$t$     is equivalently written as
\begin{equation*}
  \Pi(\vect{x} )= \argmin\limits_{\vect{z} \in \real^n } \frac{1}{2}\Norm{\vect{z}-\vect{x}}^2 + \alpha_t \ell(\vect{z}),
\end{equation*}
with $\ell(\vect{z})=0$ if $\vect{z} \in \mathcal{X}$, otherwise
$+\infty$.  Note that the projection operation is a convex problem
with the objective function being strongly convex. Thus, $\Pi(\vect{x}
)$ is a singleton (the unique minimizer) and satisfies the optimality
condition~\cite{RTR-RJBW:98}
\begin{equation*}
  \vect{x} - \Pi(\vect{x} ) \in \alpha_t \partial \ell (\Pi(\vect{x})),
\end{equation*}
where the r.h.s.~is the sub-differential set of $\ell$ at
$\Pi(\vect{x})$. Equivalently, we write the above condition as
\begin{equation*}
\Pi(\vect{x} ) = \vect{x} - \alpha_t \partial \ell (\Pi(\vect{x})).
\end{equation*}
We apply this equivalent representation to the solution
system~\eqref{eq:solsim}, resulting in
\begin{equation}
  \begin{aligned}
    \vect{x}_{t+1} =& \;  \vect{y}_{t}                                            - \alpha_t  \nabla f_{t}( \vect{y}_{t}) - \alpha_t \partial \ell (\vect{w}_{t} )     , \\
    \vect{y}_{t+1} =& \; \vect{x}_{t+1} + \beta_{t+1} \; (\vect{x}_{t+1} -\vect{x}_{t}) , \\
    \vect{w}_{t}=& \vect{x}_{t+1}. \\
  \end{aligned} \label{eq:solsimlinear}
\end{equation}
Note that~\eqref{eq:solsimlinear} is not an explicit online algorithm,
as the state $\vect{x}_{t+1}$ is yet to be determined. However, we
leverage this equivalent reformulation for the convergence analysis of
solutions to~\eqref{eq:solsim} to a sequence of optimizers
of~\eqref{eq:P2sim}, denoted by $\{ \vect{x}_t^{\star}
\}$. To do this, let $\vect{z}_{t}:=(\vect{x}_{t} -\vect{x}_{t}^{\star}, \;
\vect{x}_{t-1} -\vect{x}_{t-1}^{\star} ) $ denote the tracking error
vector and represent~\eqref{eq:solsimlinear} as the error dynamical
system
\begin{equation}
\begin{aligned}
  \vect{z}_{t+1} =& \; A_t \vect{z}_{t} + B_{t}^{\vect{u}} \vect{u}_{t} + B_{t}^{\vect{v}} \vect{v}_{t} , \\
  & \textrm{with } \vect{z}_{1}=(\vect{x}_{1} -\vect{x}_{1}^{\star},
  \; \vect{x}_{0} -\vect{x}_{0}^{\star} ),
\end{aligned} \label{eq:solerror}
\end{equation}
with the gradient input $\vect{u}_t:= \nabla f_t(\vect{y}_t)
+ \partial \ell (\vect{w}_t)$, the reference signal $\vect{v}_t:=
(\vect{x}_{t}^{\star} -\vect{x}_{t-1}^{\star}, \;
\vect{x}_{t+1}^{\star} -\vect{x}_{t}^{\star})$, the matrices
\begin{equation*}
  A_t = \begin{bmatrix} 1+\beta_t  & -\beta_t  \\ 1 & 0 \end{bmatrix}, \;
  B_{t}^{\vect{u}} = \begin{bmatrix} -\alpha_t   \\  0 \end{bmatrix}, \;
  B_{t}^{\vect{v}} = \begin{bmatrix} \beta_t  & -1   \\  0 & 0 \end{bmatrix}, \;
\end{equation*}
and the auxiliary variables
\begin{equation*}
  \begin{aligned}
    \vect{y}_t - \vect{x}_{t}^{\star}&= \begin{bmatrix} 1+\beta_t  &  -\beta_t \end{bmatrix} \vect{z}_t + \begin{bmatrix} \beta_t  &  0 \end{bmatrix} \vect{v}_{t} , \\
    \vect{w}_{t} - \vect{x}_{t}^{\star}&= \begin{bmatrix} 1 &
      0 \end{bmatrix} \vect{z}_{t+1} + \begin{bmatrix} 0 &
      1 \end{bmatrix} \vect{v}_{t}.
    \end{aligned}
\end{equation*}

We provide the following stability analysis of the system.
\begin{theorem}[Stability of~\eqref{eq:solsim}] \label{thm:solsim}
Consider the solution algorithm~\eqref{eq:solsim}, or equivalently~\eqref{eq:solsimlinear}. \\
(1) For each $t \geq 1$, we have the following
\begin{equation*}
  \begin{aligned} \small
   f_t(\vect{x}_t) -  f_t(\vect{x}_{t+1})  \geq
 \trans{\vect{\xi}_{t}}
X_{1,t} \; \vect{\xi}_{t},  \\
  \end{aligned}
\end{equation*}
\begin{equation*}
  \begin{aligned} \small
   f_t(\vect{x}_t^{\star}) -  f_t(\vect{x}_{t+1})   \geq
\trans{\vect{\xi}_{t}}
 X_{2,t}  \; \vect{\xi}_{t}.  \\
  \end{aligned}
\end{equation*}
Here, $\vect{\xi}_{t} := (\vect{z}_t , \;
\vect{u}_t, \; \vect{v}_{t}) $, and
\begin{equation*} \small
  X_{1,t}:= \frac{1}{2}
  \begin{pmatrix}
  m \beta^2 & -m \beta^2 & - \beta & m \beta^2 & 0\\ - m \beta^2 & m
  \beta^2 & \beta & -m \beta^2 & 0 \\ - \beta & \beta & \alpha (2-L
  \alpha) & - \beta & 0 \\ m \beta^2 & -m \beta^2 & - \beta & m
  \beta^2 & 0 \\ 0 & 0 & 0 & 0 & 0
\end{pmatrix},
\end{equation*}
\vspace*{-2ex}
\begin{equation*} \small
  X_{2,t}:= \frac{1}{2}
  \begin{pmatrix}
  m (1+\beta)^2 & -\eta & - (1+\beta) & \eta & 0 \\ - \eta & m \beta^2
  & \beta & -m \beta^2 & 0 \\ - (1+\beta) & \beta & \alpha (2-L
  \alpha) & - \beta & 0 \\ \eta  & -m \beta^2  & - \beta & m \beta^2 &
  0 \\ 0 & 0 & 0 & 0 & 0
\end{pmatrix},
\end{equation*}
with $\eta=m (1+\beta)\beta $ and the parameters $(m,L,\alpha,\beta)$
are a short-hand notation for $(m_t,L_t,\alpha_t,\beta_t)$. \\

\textbf{(2)} Given the horizon parameter $T_0 \in \integerpositive$
with $T=\min \{t-1, T_0 \}$. Then, for any $t \geq 2$, the solution
$\vect{x}_{t}$ from~\eqref{eq:solsim} achieves
\begin{equation*}
  \begin{aligned} \small
    &  f_{t}(\vect{x}_{t}) - f_{t}(\vect{x}_{t}^{\star})  \leq  \frac{4G_{t}}{(t+2)^2} + T F_{t} + T K_{t} \\
    & \hspace{4ex} +     \frac{4(t-T-1+\delta_{0})^2 } {(t+2)^2} ( f_{t-T}(\vect{x}_{t-T})
    - f_{t-T}(\vect{x}_{t-T}^{\star}) ) .
  \end{aligned}
\end{equation*}
where the time-dependent parameters $G_{t}$, $F_{t}$ and $K_{t}$ are
determined by $f_{t}$, $\alpha_{t}$ and $\beta_{t}$.
\end{theorem}
\vspace*{-2ex}
\begin{proof}{Theorem~\ref{thm:solsim}}
  (1) By the $m$-strong convexity and $L$-smoothness of $f$, we have
      \begin{align} \small
      f(\vect{x}) - f(\vect{y}) & \geq \trans{\nabla f(\vect{y})}
      (\vect{x} - \vect{y}) + \frac{m}{2} \Norm{\vect{x}
        -\vect{y}}^2, \label{eq:Mconvex} \\       f(\vect{y}) - f(\vect{x}) & \geq \trans{\nabla f(\vect{y})}
      (\vect{y} - \vect{x}) - \frac{L}{2} \Norm{\vect{y} -\vect{x}}^2
      . \label{eq:Lsmooth}     \end{align}
      (1a) Consider~\eqref{eq:Mconvex} with $(\vect{x}, \vect{y}) \equiv
    (\vect{x}_t, \vect{y}_t)$. We leverage $\vect{y}_{t} = \;
    \vect{x}_{t} + \beta \; (\vect{x}_{t} -\vect{x}_{t-1})$
    and     the distributive law    \footnote{Apply 1)
      $\trans{a}c=\trans{(a+b)}(c-d)+\trans{(a+b)}d-\trans{b}c$ and 2)
      $\trans{c}{c}= \trans{(c-d)}(c-d) + 2 \trans{(c-d)}d +
      \trans{d}d$, with $a=\nabla f(\vect{y}_t)$, $b=\partial \ell
      (\vect{w}_t)$, $c=\vect{x}_{t-1} - \vect{x}_t$,
      $d=\vect{x}_{t-1}^{\star} - \vect{x}_t^{\star} $,} for
  \begin{equation*}
    \begin{aligned}
      &  f(\vect{x}_t) -  f(\vect{y}_t)  \\
      & \hspace{2ex} \geq \beta \trans{\nabla f(\vect{y}_t)}
      (\vect{x}_{t-1} - \vect{x}_t) + \frac{m \beta^2}{2} \Norm{\vect{x}_{t-1} -\vect{x}_t}^2,  \\
      & \hspace{2ex} = \beta \trans{ (\nabla f(\vect{y}_t) + \partial
        \ell (\vect{w}_t) )}
      (\vect{x}_{t-1} - \vect{x}_t - \vect{x}_{t-1}^{\star} + \vect{x}_t^{\star}) \\
      & \hspace{2ex} + \frac{m \beta^2}{2} \Norm{\vect{x}_{t-1} -
        \vect{x}_t -
        \vect{x}_{t-1}^{\star} +\vect{x}_t^{\star}}^2 \\
      & \hspace{2ex} + \beta \trans{ (\nabla f(\vect{y}_t) + \partial
        \ell (\vect{w}_t) )}
      (\vect{x}_{t-1}^{\star} - \vect{x}_t^{\star}) \\
      & \hspace{2ex} - \beta \trans{\partial \ell (\vect{w}_t) } (\vect{x}_{t-1} - \vect{x}_t) \\
      & \hspace{2ex} + m \beta^2 \trans{(\vect{x}_{t-1} - \vect{x}_t -
        \vect{x}_{t-1}^{\star} +
        \vect{x}_t^{\star} ) } (\vect{x}_{t-1}^{\star} - \vect{x}_t^{\star}) \\
      & \hspace{2ex} + \frac{m \beta^2}{2} \Norm{\vect{x}_{t-1}^{\star} -\vect{x}_t^{\star}}^2. \\
                  \end{aligned}
  \end{equation*}
  We re-organize the the right-hand-side into the matrix form as
  \begin{equation*} \small
   \frac{1}{2} \trans{\vect{\delta}_{t}}
   \begin{pmatrix}
   m \beta^2 & -m \beta^2 & - \beta & m \beta^2 \\ - m \beta^2 & m \beta^2 & \beta & -m \beta^2  \\ - \beta & \beta & 0 & - \beta \\ m \beta^2 & -m \beta^2 & - \beta & m \beta^2 \end{pmatrix}  \vect{\delta}_{t}  - \beta \trans{\partial \ell (\vect{w}_t) } (\vect{x}_{t-1} - \vect{x}_t),
  \end{equation*}
                      with $\small \trans{\vect{\delta}_{t}} := ( \vect{x}_t -
  \vect{x}_t^{\star}, \vect{x}_{t-1} - \vect{x}_{t-1}^{\star}, \nabla
  f(\vect{y}_t) + \partial \ell (\vect{w}_t), \vect{x}_{t}^{\star}
  -\vect{x}_{t-1}^{\star} )$.

  (1b) Consider~\eqref{eq:Lsmooth} with $(\vect{x}, \vect{y}) \equiv
  (\vect{x}_{t+1}, \vect{y}_t)$. We leverage $\vect{x}_{t+1} =
  \vect{y}_{t} - \alpha \nabla f_{t}( \vect{y}_{t}) - \alpha \partial
  \ell (\vect{w}_{t} )$ and the distribution
  law,          resulting in
  \begin{equation*}
    \begin{aligned}
      & f(\vect{y}_t) - f(\vect{x}_{t+1}) \geq \alpha
      \trans{\nabla f(\vect{y}_t)} ( \nabla f(\vect{y}_t) +  \partial \ell (\vect{w}_t) ) \\
      & \hspace{26ex}
      - \frac{L \alpha^2}{2} \Norm{\nabla f(\vect{y}_t) +
        \partial \ell (\vect{w}_t) }^2,  \\
      & \hspace{2ex} = \frac{\alpha (2-L \alpha)}{2}  \Norm{\nabla f(\vect{y}_t) +  \partial \ell (\vect{w}_t) }^2 \\
      & \hspace{22ex} - \alpha \trans{\partial \ell (\vect{w}_t)} ( \nabla f(\vect{y}_t) +  \partial \ell (\vect{w}_t) ).  \\
    \end{aligned}
  \end{equation*}
  Now, we
    sum the terms involving $\partial \ell(\vect{w}_t) $ in the
    r.h.s.~of inequalities in (1a) and (1b),
  leverage~\eqref{eq:solsimlinear},
          and then apply the convexity of $\ell$, $\vect{x}_t \in \mathcal{X}$ and
  $\vect{w}_t=\vect{x}_{t+1} \in \mathcal{X}$, to obtain the following
  \begin{equation*}
    \begin{aligned} \tiny
      & -\beta \trans{\partial \ell (\vect{w}_t) }
      (\vect{x}_{t-1} - \vect{x}_t) - \alpha \trans{\partial \ell
        (\vect{w}_t)} ( \nabla f(\vect{y}_t) +  \partial \ell (\vect{w}_t) ) \\
      &  =- \trans{\partial \ell (\vect{w}_t) } (\vect{x}_{t} -
      \vect{w}_t)
      \geq \ell (\vect{w}_t)  - \ell (\vect{x}_t)  =0,
    \end{aligned}
  \end{equation*}
which results in $f(\vect{x}_t) -  f(\vect{x}_{t+1})  \geq
   \trans{\vect{\xi}_{t}}
  X_{1,t} \; \vect{\xi}_{t}$.

  Note that we have identified $(f, m, L, \alpha, \beta )$ with $(f_t,
  m_t, L_t, \alpha_t, \beta_t )$, and note that $\nabla
  f_t(\vect{x}^{\star}_t) + \partial \ell (\vect{x}^{\star}_t)=0$.

  (1c) Similarly, consider~\eqref{eq:Mconvex} with $(\vect{x},
  \vect{y}) \equiv (\vect{x}_t^{\star}, \vect{y}_t)$. From
  $\vect{y}_{t} = \; \vect{x}_{t} + \beta \; (\vect{x}_{t}
  -\vect{x}_{t-1})$ and the distributive law,
    \begin{equation*}
    \begin{aligned}
      &  f(\vect{x}_t^{\star}) -  f(\vect{y}_t)  \\
      & \geq \trans{\nabla f(\vect{y}_t)} (\vect{x}_{t}^{\star} -
      \vect{y}_t)
      + \frac{m }{2} \Norm{\vect{x}_{t}^{\star} -\vect{y}_t}^2,  \\
      & = \trans{ (\nabla f(\vect{y}_t) + \partial \ell (\vect{w}_t)
        )}
      (\vect{x}_{t}^{\star} - \vect{y}_{t} + \beta \vect{x}_{t}^{\star} - \beta \vect{x}_{t-1}^{\star}  ) \\
      & + \frac{m}{2} \Norm{ - (1+ \beta) (\vect{x}_{t}
        -\vect{x}_{t}^{\star} )
        + \beta ( \vect{x}_{t-1} -\vect{x}_{t-1}^{\star} ) }^2  \\
      & - \beta \trans{ (\nabla f(\vect{y}_t) + \partial \ell
        (\vect{w}_t) )}
      (\vect{x}_{t}^{\star} - \vect{x}_{t-1}^{\star})          -  \trans{\partial \ell (\vect{w}_t) } (\vect{x}_{t}^{\star} - \vect{y}_t) \\
      & - m \beta \trans{ [- (1+ \beta) (\vect{x}_{t}
        -\vect{x}_{t}^{\star} ) + \beta ( \vect{x}_{t-1}
        -\vect{x}_{t-1}^{\star} ) ] }
      (\vect{x}_{t}^{\star} - \vect{x}_{t-1}^{\star}) \\
      & + \frac{m \beta^2}{2} \Norm{\vect{x}_{t}^{\star} -\vect{x}_{t-1}^{\star}}^2, \\
      & = \frac{1}{2} \trans{\vect{\delta}_{t}}
   \begin{pmatrix}
     m (1+\beta)^2 & -\eta & - (1+\beta) & \eta \\ - \eta & m \beta^2
     & \beta & -m \beta^2 \\ - (1+\beta) & \beta & 0 & - \beta \\ \eta
     & -m \beta^2  & - \beta & m \beta^2
   \end{pmatrix}  \vect{\delta}_{t} \\
   & - \trans{\partial \ell (\vect{w}_t) } (\vect{x}_{t}^{\star} -
   \vect{y}_t) ,
    \end{aligned}
  \end{equation*}
  with $\eta=m (1+\beta)\beta $.  We add this inequality to that in
  (1b) and leverage
  \begin{equation*}
    \begin{aligned} \tiny
      & -  \trans{\partial \ell (\vect{w}_t) }
      (\vect{x}_{t}^{\star} - \vect{y}_t)
      - \alpha \trans{\partial \ell (\vect{w}_t)} ( \nabla
      f(\vect{y}_t) +
      \partial \ell (\vect{w}_t) ) \\
      & =- \trans{\partial \ell (\vect{w}_t) } (\vect{x}_{t}^{\star} -
      \vect{w}_t) \geq \ell (\vect{w}_t) - \ell (\vect{x}_t^{\star})
      =0,
    \end{aligned}
  \end{equation*}
  resulting in
        $  f(\vect{x}_t^{\star}) - f(\vect{x}_{t+1}) \geq
      \trans{\vect{\xi}_{t}}
      X_{2,t} \; \vect{\xi}_{t}. $

(2) Let us define the time varying function
\begin{equation*}
  V_t(\vect{z}_t):=  \trans{\begin{bmatrix} \vect{z}_t \\
      \vect{x}_{t}^{\star} -\vect{x}_{t-1}^{\star}
    \end{bmatrix}} H_t \begin{bmatrix} \vect{z}_t \\
    \vect{x}_{t}^{\star} -\vect{x}_{t-1}^{\star}
  \end{bmatrix},
\end{equation*}
where we take
\begin{equation*}
  H_t:= \frac{1}{2 \alpha_{t-1}}
  \begin{bmatrix}
    \delta_{t-1} \\ 1- \delta_{t-1} \\ \delta_{t-1}
  \end{bmatrix}
  \begin{bmatrix}
    \delta_{t-1}, & 1- \delta_{t-1}, & \delta_{t-1}
  \end{bmatrix},
\end{equation*}
with $\{ \alpha_t\}_t$ those in the solution system~\eqref{eq:solsim} and $\{ \delta_t \}_t$ the sequence of scalars which defines $\{\beta_t\}_t$. Now, verify
\begin{equation*}
  \begin{aligned}
    &V_{t+1}(\vect{z}_{t+1}) - \frac{\alpha_{t-1}}{\alpha_{t}}
    V_t(\vect{z}_t) = \trans{\vect{\xi}_{t}} J_{t} \vect{\xi}_{t},
  \end{aligned}
\end{equation*}
where $\vect{\xi}_{t} := (\vect{z}_t , \; \vect{u}_t, \; \vect{v}_{t})$, which are those define~\eqref{eq:solerror}, resulting in
$\vect{\xi}_{t} := (\vect{x}_{t} -\vect{x}_{t}^{\star}, \;
\vect{x}_{t-1} -\vect{x}_{t-1}^{\star}, \nabla f_t(\vect{y}_t)
+ \partial \ell (\vect{w}_t), \vect{x}_{t}^{\star} -\vect{x}_{t-1}^{\star}, \;
\vect{x}_{t+1}^{\star} -\vect{x}_{t}^{\star} )$
and
  \begin{equation*}
  J_{t} = \frac{1}{2 \alpha_{t}}
  \begin{pmatrix}
    0 & 0 & - \alpha_t \delta_{t} \delta_{t-1} & - \delta_{t-1}  & 0\\
    0 & 0 & \alpha_t \beta_{t} \delta_{t}^2  & \beta_{t} \delta_{t} & 0 \\
    - \alpha_t \delta_{t}\delta_{t-1} & \alpha_t \beta_{t}
    \delta_{t}^2 & \alpha_t^2 \delta_{t}^2 & - \alpha_t \beta_{t}
    \delta_{t}^2
    & 0 \\
    - \delta_{t-1} & \beta_{t} \delta_{t} & -
    \alpha_t \beta_{t} \delta_{t}^2 & 1-  2 \delta_{t-1} & 0 \\
    0 & 0 & 0 & 0 & 0 \end{pmatrix}.
\end{equation*}
Let us compute
\begin{equation*} \tiny
\begin{aligned}
  & M_t:= \delta_{t-1}^2 X_{1,t}
  + \delta_t X_{2,t}  \\
  & = \frac{1}{2} \begin{pmatrix} m_t(\delta_{t}^2 -1) & -m_{t}
    \beta_{t} \delta_{t}\delta_{t-1}
    & -\delta_{t}\delta_{t-1}  & m_{t} \beta_{t} \delta_{t}\delta_{t-1} & 0\\
    -m_{t} \beta_{t} \delta_{t}\delta_{t-1} & m_t
    \beta_{t}^2\delta_{t}^2
    & \beta_{t}\delta_{t}^2 & -m_{t} \beta_{t}^2\delta_{t}^2 & 0 \\
    - \delta_{t}\delta_{t-1} & \beta_{t}\delta_{t}^2 &
    \alpha_{t}(2-L_t \alpha_{t}) \delta_{t}^2
    & - \beta_{t}\delta_{t}^2 & 0 \\
    m_{t} \beta_{t} \delta_{t}\delta_{t-1}
    & -m_{t} \beta_{t}^2\delta_{t}^2 & - \beta_{t}\delta_{t}^2 & m_{t} \beta_{t}^2 \delta_{t}^2 & 0 \\
    0 & 0 & 0 & 0 & 0 \end{pmatrix},
\end{aligned}
\end{equation*}
and then achieve
\begin{equation*}
  \begin{aligned}
    \trans{\vect{\xi}_{t}  } (J_{t} - M_{t})  \vect{\xi}_{t}
    &= \trans{\begin{bmatrix} \vect{z}_{t} \\ \vect{x}_{t}^{\star}
        -\vect{x}_{t-1}^{\star}
      \end{bmatrix} } N_{1,t}  \begin{bmatrix} \vect{z}_{t}
      \\ \vect{x}_{t}^{\star} -\vect{x}_{t-1}^{\star} \end{bmatrix} \\
    &  + \trans{\begin{bmatrix} \vect{z}_{t} \\ \vect{x}_{t}^{\star}
        -\vect{x}_{t-1}^{\star}
      \end{bmatrix} } N_{2,t}
    \begin{bmatrix} \vect{z}_{t} \\ \vect{x}_{t}^{\star}
      -\vect{x}_{t-1}^{\star}
    \end{bmatrix}  \\
    & - \alpha_{t}(1-L_{t} \alpha_{t})
    \trans{\vect{u}_{t}}\vect{u}_{t},
  \end{aligned}
\end{equation*}
with, for each $t\geq 1$,
\begin{equation*}
\begin{aligned}
  N_{1,t} &:= \frac{1}{2}
  \begin{pmatrix}
  -  m_t(\delta_{t}^2 -1)   & m_{t} \beta_{t} \delta_{t}\delta_{t-1} & -m_{t} \beta_{t} \delta_{t}\delta_{t-1} \\
    m_{t} \beta_{t} \delta_{t}\delta_{t-1}  & -m_{t} \beta_{t}^2\delta_{t}^2 & m_{t} \beta_{t}^2\delta_{t}^2 \\
    -m_{t} \beta_{t} \delta_{t}\delta_{t-1} & m_{t} \beta_{t}^2\delta_{t}^2 & -m_{t} \beta_{t}^2\delta_{t}^2
  \end{pmatrix},   \\
  & \cong \frac{m_{t}}{2}
  \begin{pmatrix}
  -  (\delta_{t}^2 -1)   &  \beta_{t} \delta_{t}\delta_{t-1} & 0 \\
     \beta_{t} \delta_{t}\delta_{t-1}  & - \beta_{t}^2\delta_{t}^2 & 0 \\
    0 & 0 & 0
  \end{pmatrix} \preceq 0,
\end{aligned}
\end{equation*}
and, using the fact that $\delta_{t}> (t+1)/2$, $\forall \; t \geq 0$, we have
\begin{equation*}
  N_{2,t}:= \frac{1}{2} \begin{pmatrix}
 0 & 0 & - \delta_{t-1} \\
  0 & 0 & \beta_{t}\delta_{t} \\
    -\delta_{t-1} &  \beta_{t}\delta_{t} & 1-2 \delta_{t-1} &
\end{pmatrix}  \preceq 0.
\end{equation*}
Then, if we select $\alpha_{t} \leq 1/L_{t}$, it results in
\begin{equation*}
    \trans{\vect{\xi}_{t}  } (J_{t} - M_{t})  \vect{\xi}_{t}  \leq 0.
\end{equation*}
We rewrite it as
\begin{equation*}
  \begin{aligned}
    &  V_{t+1}(\vect{z}_{t+1}) -
    \frac{\alpha_{t-1}}{\alpha_{t}} V_t(\vect{z}_t) \leq  \trans{\vect{\xi}_{t}  }  M_{t}  \vect{\xi}_{t} , \\
    & \leq   \delta_{t-1}^2 ( f_t(\vect{x}_t) - f_t(\vect{x}_{t+1}) )
    +
    \delta_t ( f_t(\vect{x}_t^{\star})-  f_t(\vect{x}_{t+1})), \\
    & = - \delta_{t}^2( f_t(\vect{x}_{t+1}) - f_t(\vect{x}_t^{\star})
    ) + \delta_{t-1}^2 ( f_t(\vect{x}_t) - f_t(\vect{x}_t^{\star}) ).
  \end{aligned}
\end{equation*}
As $f_t$ being locally Lipschitz in $t$, there exists a non-negative
function $h(\vect{x})$ such that
\begin{equation*}
  f_{t+1}(\vect{x}_{t+1}) - f_t(\vect{x}_{t+1}) \leq h(\vect{x}_{t+1}),
\end{equation*}
resulting in
\begin{equation*}
  \begin{aligned}
&  V_{t+1}(\vect{z}_{t+1}) -
\frac{\alpha_{t-1}}{\alpha_{t}} V_t(\vect{z}_t)  \\
& \leq -  \delta_{t}^2( f_{t+1}(\vect{x}_{t+1}) -
f_{t+1}(\vect{x}_{t+1}^{\star})  )  +   \delta_{t-1}^2 ( f_t(\vect{x}_t) - f_t(\vect{x}_t^{\star}) ) \\
& -  \delta_{t}^2( f_{t+1}(\vect{x}_{t+1}^{\star}) -
f_t(\vect{x}_t^{\star})  ) + \delta_{t}^2 h(\vect{x}_{t+1}), \;
\forall\, t
  \end{aligned}
\end{equation*}
Summing up the above set of inequalities over the moving horizon
window $t \in \mathcal{T}=\untilinterval{t-1}{t-T}$, where $T=\min
\{t-1, T_0 \}$ with some $T_0 \in \integerpositive$, we obtain
\begin{equation*}
  \begin{aligned}
    & V_{t}(\vect{z}_{t}) + \sum_{k\in \mathcal{T}} (1-
    \frac{\alpha_{k-1}}
    {\alpha_{k}}) V_k(\vect{z}_k) -  V_{t-T}(\vect{z}_{t-T}) \\
    & \leq -  \delta_{t-1}^2( f_{t}(\vect{x}_{t}) - f_{t}(\vect{x}_{t}^{\star})  ) \\
    & \quad  +   \delta_{t-T-1}^2 ( f_{t-T}(\vect{x}_{t-T}) - f_{t-T}(\vect{x}_{t-T}^{\star}) ) \\
    & \quad - \sum_{k\in \mathcal{T}} \delta_{k}^2(
    f_{k+1}(\vect{x}_{k+1}^{\star}) - f_k(\vect{x}_k^{\star}) ) +
    \sum_{k \in \mathcal{T}} \delta_{k}^2 h(\vect{x}_{k+1}).
  \end{aligned}
\end{equation*}
Let us denote by $G_t$, $K_t$, and $F_t$, respectively, the horizon
accumulated potential, the bound of the locally Lipschitz function
$h$, and the variation bound of the optimal objective values. That is,
\begin{equation*}
  \begin{aligned}
    G_{t}:=& V_{t-T}(\vect{z}_{t-T})
    - V_{t}(\vect{z}_{t}) - \sum_{k\in \mathcal{T} } (1-  \frac{\alpha_{k-1}}{\alpha_{k}}) V_k(\vect{z}_k)   , \\
    K_{t}:=& \max_{k \in \mathcal{T}} \left\{ h(\vect{x}_{k+1}) \right\} , \\
    F_{t}:=& \max_{k \in \in \mathcal{T}} \left\{ |
      f_{k+1}(\vect{x}_{k+1}^{\star}) - f_k(\vect{x}_k^{\star}) |
    \right\}.
  \end{aligned}
\end{equation*}
Then, using the fact that \textbf{(1)} $\delta_{t-1} \geq (t+2)/2$,
for all $t\geq 0$; \textbf{(2)} $\delta_{t-T-1} \leq t-T-1 +
\delta_{0}$ with $\delta_{0}=(1+ \sqrt{5})/2$, and \textbf{(3)}
$\delta_{t}$ is monotonically increasing, we have
\begin{equation*}
  \begin{aligned}
&  f_{t}(\vect{x}_{t}) - f_{t}(\vect{x}_{t}^{\star})  \leq  \frac{4G_{t}}{(t+2)^2} + T F_{t} + T K_{t} \\
  & \hspace{4ex} +   \frac{4(t-T-1+\delta_{0})^2 } {(t+2)^2} ( f_{t-T}(\vect{x}_{t-T}) - f_{t-T}(\vect{x}_{t-T}^{\star}) ) .
  \end{aligned}
\end{equation*}
Note that, when $t \leq T_0+1$, we have $T=t-1$. This gives
\begin{equation*}
  \begin{aligned}
    & f_{t}(\vect{x}_{t}) - f_{t}(\vect{x}_{t}^{\star}) \leq
    \frac{4G_{t}}{(t+2)^2} +
    (t-1) F_{t} + (t-1) K_{t} \\
    & \hspace{20ex} +     \frac{4\delta_{0}^2 } {(t+2)^2} ( f_{1}(\vect{x}_{1}) -
    f_{1}(\vect{x}_{1}^{\star}) ) . \hfill \textrm{\qed}
  \end{aligned}
\end{equation*}
  \begin{comment}
\end{comment}
\end{proof}

\vspace*{-2ex}
\subsection{Proofs of lemmas and theorems}
\vspace*{-2ex}
 \begin{proof}{Lemma~\ref{lemma:Reform}}
            By the definition of the ambiguity set, we have that, for
   any distribution $\probQ \in
   \mathcal{P}_{t+1}(\vect{\alpha},\vect{u})$
 \begin{equation*}
   \subscr{d}{W}(\probQ,\hat{\probP}_{t+1|t}) \leq  \hat{\epsilon}    ,
 \end{equation*}
 which, by Kantorovich-Rubinstein Theorem, is equivalent to
 \begin{equation*}
   \int_{\mathcal{Z}} {h(\vect{x})\probQ(d{\vect{x}})}
   -\int_{\mathcal{Z}} {h(\vect{x})\hat{\probP}_{t+1|t}(d{\vect{x}})} \leq  \hat{\epsilon}    , \quad \forall h \in
     \mathcal{L},
 \end{equation*}
 where $\mathcal{L}$ is the set of functions with Lipschitz constant $1$ and $\mathcal{Z}$ is the support of the random variable $\vect{x}$. For a given $\vect{u}$, let us select $h$ to be
 \begin{equation*}
 h(\vect{x}):=\frac{\ell(\vect{u}, \vect{x})}{L(\vect{u})},
 \end{equation*}
 where $L$ is the positive Lipschitz function as in~Assumption~\ref{assump:cvxloss}. Substituting $h$ to the above inequality, we have
 \begin{equation*}
   \begin{aligned}
&    \int_{\mathcal{Z}} {\ell(\vect{u}, \vect{x})\probQ(d{\vect{x}})}
   -\int_{\mathcal{Z}} {\ell(\vect{u}, \vect{x})\hat{\probP}_{t+1|t}(d{\vect{x}})}
       \leq  \hat{\epsilon}    L(\vect{u}),
   \end{aligned}
 \end{equation*}
 or equivalently
 \begin{equation*}
   \begin{aligned}
     \mathbb{E}_{\probQ} \left[ \ell(\vect{u}, \vect{x})\right]   \leq \mathbb{E}_{\hat{\probP}_{t+1|t}(\vect{\alpha},\vect{u})} \left[ \ell(\vect{u}, \vect{x})\right] + \hat{\epsilon}      L(\vect{u}).
   \end{aligned}
 \end{equation*}
 As the inequality holds for every $\probQ \in \mathcal{P}_{t+1}$, therefore
 \begin{equation*}
   \begin{aligned}
 & \sup\limits_{\probQ \in \mathcal{P}_{t+1}(\vect{\alpha},\vect{u})}  \; \mathbb{E}_{\probQ} \left[ \ell(\vect{u}, \vect{x})\right]  \\
 & \hspace{8ex} \leq \mathbb{E}_{\hat{\probP}_{t+1|t}(\vect{\alpha},\vect{u})} \left[ \ell(\vect{u}, \vect{x})\right] + \hat{\epsilon}(t,T,\beta,\vect{\alpha},\vect{u}) L(\vect{u}). \hfill \textrm{\qed}
   \end{aligned}
 \end{equation*}
                  \end{proof}

 \begin{proof}{Theorem~\ref{thm:Reform}}
   We show this by constructing a distribution in the ambiguity
   set. By Assumption~\ref{assump:gradient} on convex and
   gradient-accessible functions, there exists an
   index    $j \in \mathcal{T}$ such that the derivative $\nabla_{\vect{x}}
   \ell(\vect{u}, \vect{x})$ at $(\vect{u}, \bar{\vect{x}}^{(j)})$,
   $\bar{\vect{x}}^{(j)}:=\sum_{i=1}^{p} \alpha_i
   \xi_{j}^{(i)}(\vect{\alpha},\vect{u})$, satisfies
 \begin{equation*}
  \Norm{\nabla_{\vect{x}} \ell(\vect{u}, \bar{\vect{x}}^{(j)})}  =L(\vect{u}).
 \end{equation*}
 Now using this index $j$, we construct a parameterized distribution
 as follows
 \begin{equation*}
   \probQ(\Delta \vect{x}) = \frac{1}{T} \sum\limits_{k\in \mathcal{T}, k \ne j} \delta_{\{
     \sum\limits_{i=1}^{p} \alpha_i \xi_k^{(i)}(\vect{\alpha},\vect{u}) \}} + \frac{1}{T} \delta_{\{
     \bar{\vect{x}}^{(j)} +\Delta \vect{x} \}},
\end{equation*}
where $\Delta \vect{x} \in \real^n$ with $\Norm{\Delta \vect{x}} \leq
T \hat{\epsilon} $. By the definition of the ambiguity set and, since the support of
the distribution $\probP$ is $\Xi_{t+1}=\real^n$, we have $\probQ(\Delta
\vect{x}) \in \mathcal{P}_{t+1}(\vect{\alpha}, {\vect{u}} )$.

Next, we quantify the lower bound of the following term  \begin{equation*}
   \begin{aligned}
     & \mathbb{E}_{\probQ(\Delta \vect{x})} \left[ \ell(\vect{u},
       \vect{x})\right]
     -  \mathbb{E}_{\hat{\probP}_{t+1|t}(\vect{\alpha},\vect{u})} \left[ \ell(\vect{u}, \vect{x})\right] \\
     & \hspace{3cm} = \frac{1}{T} \left( \ell(\vect{u},
       \bar{\vect{x}}^{(j)} + \Delta \vect{x} ) - \ell(\vect{u},
       \bar{\vect{x}}^{(j)} ) \right).
   \end{aligned}
 \end{equation*}
 By Assumption~\ref{assump:gradient} on the convexity of $\ell$ on
 $\vect{x}$, we have
 \begin{equation*}
   \begin{aligned}
     \ell(\vect{u}, \bar{\vect{x}}^{(j)} + \Delta \vect{x} ) & -
   \ell(\vect{u}, \bar{\vect{x}}^{(j)} )                \geq \trans{\nabla_{\vect{x}}\ell(\vect{u}, \bar{\vect{x}}^{(j)})} \Delta \vect{x}. \\
      \end{aligned}
 \end{equation*}
Then, by selecting
\begin{equation*}
  \Delta \vect{x}:= \frac{T \hat{\epsilon}   \nabla_{\vect{x}} \ell(\vect{u}, \bar{\vect{x}}^{(j)})
  }{\Norm{\nabla_{\vect{x}} \ell(\vect{u}, \bar{\vect{x}}^{(j)})} } ,
\end{equation*}
we have
\begin{equation*}
  \trans{\nabla_{\vect{x}}\ell(\vect{u}, \bar{\vect{x}}^{(j)})} \Delta \vect{x} =T  \hat{\epsilon}   L(\vect{u}).
\end{equation*}
These bounds result in
\begin{equation*}
  \begin{aligned}
& \mathbb{E}_{\probQ(\Delta \vect{x})} \left[ \ell(\vect{u}, \vect{x})\right] -  \mathbb{E}_{\hat{\probP}_{t+1|t}(\vect{\alpha},\vect{u})} \left[ \ell(\vect{u}, \vect{x})\right]
\geq  \hat{\epsilon}  L(\vect{u}).
  \end{aligned}
\end{equation*}
As $\probQ(\Delta \vect{x}) \in \mathcal{P}_{t+1}(\vect{\alpha},
{\vect{u}} )$, therefore
 \begin{equation*}
   \begin{aligned}
     & \sup\limits_{\probQ \in
       \mathcal{P}_{t+1}(\vect{\alpha},\vect{u})} \;
     \mathbb{E}_{\probQ} \left[ \ell(\vect{u}, \vect{x})\right]
  \geq \mathbb{E}_{\hat{\probP}_{t+1|t}(\vect{\alpha},\vect{u})} \left[
   \ell(\vect{u}, \vect{x})\right] +
 \hat{\epsilon}  L(\vect{u}).
   \end{aligned}
 \end{equation*}
 Finally, with Assumption~\ref{assump:cvxloss} on Lipschitz loss functions and
 Lemma~\ref{lemma:Reform} on an upper bound of~\eqref{eq:P1}, we equivalently write Problem~\eqref{eq:P1}
 as
 \begin{equation*}
   \begin{aligned}
     \inf\limits_{\vect{u} \in \mathcal{U}}
     \mathbb{E}_{\hat{\probP}_{t+1|t}(\vect{\alpha},\vect{u})} \left[
       \ell(\vect{u}, \vect{x})\right] +
     \hat{\epsilon}(t,T,\beta,\vect{\alpha},\vect{u}) L(\vect{u}),
   \end{aligned}
 \end{equation*}
 which is the Problem~\eqref{eq:P2}.
 $\hfill$ \qed \end{proof}

 \begin{proof}{Lemma~\ref{lemma:Lip}}
   This is the direct application of the definition of the local
   Lipschitz condition.  $\hfill$ \qed \end{proof}

 \begin{proof}{Lemma~\ref{lemma:moreau}}
   First, we have
 \begin{equation*}
   F_{\mu}(\vect{u})
   \leq F(\vect{u}) + \frac{1}{2\mu} \Norm{\vect{u}- \vect{u} }^2
   = F(\vect{u}), \; \forall \; \vect{u} \in \mathcal{U}.
 \end{equation*}
 Then, we compute
 \begin{equation*}
   \begin{aligned}
     F(\vect{u}) -F_{\mu}(\vect{u}) &= \sup\limits_{\vect{z} \in
       \mathcal{U}} \left\{ F(\vect{u}) - F(\vect{z}) -
       \frac{1}{2\mu} \Norm{\vect{z}- \vect{u} }^2   \right\} , \\
     &\leq \sup\limits_{\vect{z} \in \mathcal{U}} \left\{
       \trans{\vect{g}(\vect{u}) }
       (\vect{u} -\vect{z}) - \frac{1}{2\mu} \Norm{\vect{z}- \vect{u} }^2   \right\}, \\
    & \leq \sup\limits_{\vect{z} } \left\{
       \trans{\vect{g}(\vect{u}) }
       (\vect{u} -\vect{z}) - \frac{1}{2\mu} \Norm{\vect{z}- \vect{u} }^2   \right\},  \\
     &\leq \frac{\mu}{2} \trans{\vect{g}(\vect{u}) } \vect{g}(\vect{u})
     \leq \frac{D}{2} \mu,
   \end{aligned}
 \end{equation*}
 where the equality comes from the definition of $F_{\mu}(\vect{u})$, the first inequality leverages the convexity of $F$, the second one relaxes the constraint set, the third one applies the achieved optimizer $\vect{z}^{\star}=\vect{u}-\mu \vect{g}(\vect{u})$, and the last one is from the boundedness of subgradients.

Further, given $F$ as described, it is
 well-known (see, e.g.,~\cite[Proposition 12.15]{HHB-PLC:11} for
 details) that $F_{\mu}$ is convex and continuously differentiable
 where its gradient $\nabla F_{\mu}$ is Lipschitz continuous with
 constant $1/\mu$. In addition, the minimizer
 $\vect{z}^{\star}(\vect{u})$ of $F_{\mu}$ is achievable and unique,
 resulting in an explicit gradient expression of $F_{\mu}$ as follows
 \begin{equation*}
  \nabla F_{\mu} (\vect{u}) =  \frac{1}{\mu}(\vect{u} -  \vect{z}^{\star}(\vect{u})  ).
 \end{equation*}
   In addition, we claim that, if $F$ is $M$-strongly convex, $F_{\mu}$
 is ${M}/{(1+ \mu M)}$-strongly convex, following~\cite[Theorem
 2.2]{CL-CS:97}.  Finally, we equivalently write the minimization
 problem as follows \begin{equation*}
  \begin{aligned}
    \min_{\vect{u} \in \mathcal{U}} F_{\mu}(\vect{u}) &=
    \min_{\vect{u} \in \mathcal{U}} \min\limits_{\vect{z} \in
      \mathcal{U}}
    \left\{ F(\vect{z}) + \frac{1}{2\mu} \Norm{\vect{z}- \vect{u} }^2   \right\} \\
    &= \min_{\vect{z} \in \mathcal{U}}
    \min\limits_{\vect{u} \in \mathcal{U}} \left\{ F(\vect{z}) + \frac{1}{2\mu} \Norm{\vect{z}- \vect{u} }^2   \right\} \\
    &= \min_{\vect{z} \in \mathcal{U}}  F(\vect{z}), \\
  \end{aligned}
\end{equation*}
where the first line applies the achievability of the minimizer of the
problem that defines $F_{\mu}$, the second switches the
minimization operators, the third applies the fact that $\vect{u}=\vect{z}$ solves the
inner problem. This concludes that any $\vect{u}$ that
minimizes $F_{\mu}$ also minimizes $F$, and vice versa.  $\hfill$
\qed \end{proof}

   \begin{proof}{Theorem~\ref{thm:regret}}
   Let us consider the solution system~\eqref{eq:oagsol}. At each time $t$, let us select $\varepsilon:=\varepsilon_t=1/ \Lip(G_{\mu})$, or equivalently, $\mu/b$ with $b = \max_{k\in\mathcal{T}} b_k $. Let $\eta_t$ satisfy
   \begin{equation*}
     \delta_{-1}=1, \; \delta_{t+1}:=\frac{1+ \sqrt{1+4\delta_t^2}}{2}, \; \eta_t:=\frac{\delta_{t-1}-1}{\delta_t}.
   \end{equation*}
   Then, by Theorem~\ref{thm:solsim} with $t \geq 2$, the following holds
   \begin{equation}
     \begin{aligned}
   G_{\mu}(t,\vect{u}_{t} ) - G_{\mu}(t,{\vect{u}}_{t}^{\star} )
    \leq & \frac{4W_{t}}{(t+2)^2} + T F_{t},
  \end{aligned} \label{eq:boundG}
   \end{equation}
where ${\vect{u}}_{t}^{\star}$ is a solution to~\eqref{eq:P2smooth}, $T=\min \{t-1, T_0 \}$ with some horizon parameter $T_0 \in \integerpositive$. Notice that $T_0$ is the length of the used historical data whenever such data are available. The time-varying parameter $W_{t}$ depends on the initial objective discrepancy and the accumulated energy storage in the considered time horizon $\mathcal{T}$, and $F_t$ is the variation bound of the optimal objective values in $\mathcal{T}$. Specifically, we have
\begin{equation*}
  \begin{aligned}
    F_{t}=& \max_{k\in \mathcal{T}} \left\{ | G_{\mu}(k+1,{\vect{u}}_{k+1}^{\star} ) - G_{\mu}(k,{\vect{u}}_{k}^{\star} )  |    \right\} + \bar{L} ,
  \end{aligned}
\end{equation*}
with $\bar{L}$ the variation bound of $G_{\mu}(t,\vect{u}_{t} )$ w.r.t. time $t$. Let us consider the storage function  $V_t(\vect{z}_t):=  \trans{\vect{z}_{t}} H_t \vect{z}_{t}$, where $\vect{z}_{t}:=(\vect{u}_{t} -\vect{u}_{t}^{\star}, \; \vect{u}_{t-1} -\vect{u}_{t-1}^{\star} , \; \vect{u}_{t}^{\star} -\vect{u}_{t-1}^{\star} ) $ and
\begin{equation*}
  H_t:= \frac{1}{2 \varepsilon_{t-1}} \begin{bmatrix} \delta_{t-1} \\ 1- \delta_{t-1} \\ \delta_{t-1}
  \end{bmatrix} \begin{bmatrix} \delta_{t-1} & 1- \delta_{t-1} & \delta_{t-1}
  \end{bmatrix} \succeq 0.
\end{equation*}
Then we have
\begin{equation*}
  \begin{aligned}
  &  W_{t}= V_{t-T}(\vect{z}_{t-T}) - V_{t}(\vect{z}_{t}) - \sum_{k\in \mathcal{T} } (1-  \frac{\varepsilon_{k-1}}{\varepsilon_{k}}) V_k(\vect{z}_k) \\
    & \hspace{8ex} + (t-T-1+\delta_{0})^2  ( f_{t-T}(\vect{x}_{t-T}) - f_{t-T}(\vect{x}_{t-T}^{\star}) ),
  \end{aligned}
\end{equation*}
where the first two term is the energy decrease in the horizon $\mathcal{T}$; the third sum term indicates the instantaneous energy change, which depends on the online, estimated Lipschitz constant; the last term depends on the goodness of the initial decision at the beginning of the current $\mathcal{T}$. Note how the selection of $\varepsilon_t$ and $T$ affect $W_t$ (or $G_t$ in Theorem~\ref{thm:solsim}). In the most conservative scenario, we select $\varepsilon_t:= \min \{ \varepsilon_{t-1}, \mu/b_t  \}$ and $T_0=\infty$, which results in a constant upper bound of $W_{t}$ as follows
\begin{equation*}
  \begin{aligned}
  &  W_{t} \leq V_{1}(\vect{z}_{1})  + \delta_{0}^2  ( f_{1}(\vect{x}_{1}) - f_{1}(\vect{x}_{1}^{\star}) ),
  \end{aligned}
\end{equation*}
therefore, in this case, the bound~\eqref{eq:boundG} essentially depends on the growing term $(t-1)F_t$. A less conservative way is to use moving horizon strategy, with $\varepsilon_t:= \min \{ \varepsilon_{t-1}, \mu/b_t  \}$ but a finite $T_0$. Then, as $t$ is sufficiently large, we have
\begin{equation*}
  \begin{aligned}
  &  W_{t} \leq V_{t-T}(\vect{z}_{t-T})  + t^2  ( f_{t-T}(\vect{x}_{t-T}) - f_{t-T}(\vect{x}_{t-T}^{\star}) ),
  \end{aligned}
\end{equation*}
where, in this case, the bound~\eqref{eq:boundG} essentially depends on $F_t$ and $f_{t-T}(\vect{x}_{t-T}) - f_{t-T}(\vect{x}_{t-T}^{\star})$.

   Now, we consider for any $t\geq 2$. By Definition~\ref{def:smooth}, there exists a constant $a>0$ such that
   \begin{equation*}
     G(t,{\vect{u}}_{t} ) -a \mu \leq G_{\mu}(t,{\vect{u}}_{t} ),    \end{equation*}
   and by Lemma~\ref{lemma:moreau}, we have that $\vect{u}_{t}^{\star}$ is a minimizer of~\eqref{eq:P2smooth} if and only if it is that of~\eqref{eq:P2}, and
   \begin{equation*}
     G_{\mu}(t,{\vect{u}}_{t}^{\star} ) \equiv G(t,{\vect{u}}_{t}^{\star} ).    \end{equation*}
      This results in
   \begin{equation}
     \begin{aligned}
   G(t,\vect{u}_{t} ) - G(t,{\vect{u}}_{t}^{\star} )
    \leq & \frac{4W_{t}}{(t+2)^2} + T F_{t} + a \mu ,
     \end{aligned}
     \label{eq:bound_P2}
   \end{equation}
with an equivalent expression of $F_{t}$ as
\begin{equation*}
  \begin{aligned}
    F_{t}=& \max_{k\in \mathcal{T}} \left\{ | G_{k+1}^{\star} - G_{k}^{\star}   |    \right\} + \bar{L} ,
  \end{aligned}
\end{equation*}
where $G_k^{\star}:= G(k,{\vect{u}}_{k}^{\star} )$ is the optimal objective value of~\eqref{eq:P2} or, later we see, equivalent to that of~\eqref{eq:P1}.

   Next, by Theorem~\ref{thm:Reform} on the equivalence of~\eqref{eq:P1} and~\eqref{eq:P2}, $\vect{u}_{t}^{\star}$ is a minimizer of~\eqref{eq:P2} if and only if it is also that of~\eqref{eq:P1}, and
   \begin{equation} \label{eq:sub_1}
    G(t,{\vect{u}}_{t}^{\star} ) \equiv \sup\limits_{\probQ \in \mathcal{P}_{t+1}(\vect{\alpha},\vect{u}_{t}^{\star})}  \; \mathbb{E}_{\probQ} \left[ \ell(\vect{u}_{t}^{\star}, \vect{x})\right].
   \end{equation}
                           Further, as in Section~\ref{sec:tractable}, we claim that Problem~\eqref{eq:P1} provides a probabilistic bound for the objective of~\eqref{eq:P}, resulting in
   \begin{align}
       &  {\Prob}\left( \probP_{t+1|t} \in
      \mathcal{P}_{t+1}     \right) \geq \rho(t), \textrm{ or equivalently, }  \label{eq:sub_intermediate1} \\
       & \Prob \left(  \mathbb{E}_{\probP_{t+1|t}} \left[ \ell(\vect{u}_t, \vect{x})\right] \leq G(t,\vect{u}_{t} )  \right) \geq \rho(t), \label{eq:sub_intermediate2}
   \end{align}
   with $\rho(t)$ as in Theorem~\ref{thm:ambiguityset}.
                 Then by~\eqref{eq:sub_intermediate1}, we know that $\probP_{t+1|t} \in \mathcal{P}_{t+1}$ if and only if $\subscr{d}{W}(\probP_{t+1|t},\hat{\probP}_{t+1|t}) \leq  \hat{\epsilon}$ where $\hat{\epsilon}$ is selected as in Theorem~\ref{thm:ambiguityset}.
      Further, since $\subscr{d}{W}$ is a metric, for any $\probQ \in \mathcal{P}_{t+1}$, we claim
   \begin{equation*}
   \begin{aligned}
   \subscr{d}{W}(\probQ, \probP_{t+1|t})  & \leq \subscr{d}{W}(\probQ, \hat{\probP}_{t+1|t}) + \subscr{d}{W}(\probP_{t+1|t}, \hat{\probP}_{t+1|t}), \\
   & \leq \hat{\epsilon}+ \hat{\epsilon} \leq 2 \hat{\epsilon}.
   \end{aligned}
   \end{equation*}
   By Assumption~\ref{assump:cvxloss} and the same proof procedure of Lemma~\ref{lemma:Reform} on the above inequality, we have, for every $\vect{u}$, the following:
   \begin{equation*}
     \sup\limits_{\probQ \in \mathcal{P}_{t+1}(\vect{\alpha},\vect{u})}  \; \mathbb{E}_{\probQ} \left[ \ell(\vect{u}, \vect{x})\right] \leq \mathbb{E}_{\probP_{t+1|t}} \left[ \ell({\vect{u}}, \vect{x})\right]   + 2 L({\vect{u}}) \hat{\epsilon}.
   \end{equation*}
By taking $\vect{u}:=\vect{u}_{t}^{\star}$ and using~\eqref{eq:sub_1}, we have
   \begin{equation} \label{eq:sub_2}
    G(t,{\vect{u}}_{t}^{\star} ) \leq \mathbb{E}_{\probP_{t+1|t}} \left[ \ell({\vect{u}_{t}^{\star}}, \vect{x})\right]   + 2L({\vect{u}_{t}^{\star}}) \hat{\epsilon}.
   \end{equation}
   We combine the inequality~\eqref{eq:bound_P2},~\eqref{eq:sub_intermediate2} and~\eqref{eq:sub_2}, resulting in
   \begin{equation*}
     \begin{aligned}
            &  \mathbb{E}_{\probP_{t+1|t}} \left[ \ell(\vect{u}_t, \vect{x})\right] - \mathbb{E}_{\probP_{t+1|t}} \left[ \ell({\vect{u}}_{t}^{\star}, \vect{x})\right]
                \\
   &    \hspace{2cm} \leq \frac{4W_{t}}{(t+2)^2} + T F_t + a \mu + 2 L({\vect{u}}_{t}^{\star}) \hat{\epsilon},
                 \end{aligned}
   \end{equation*}
   with the probability at least $\rho(t)$, holds for any $t \geq 2$.
      Furthermore, if all historical data are assimilated for the decision $\vect{u}_{t}$, i.e., we select $T_0=\infty$ with $\varepsilon_t:= \min \{ \varepsilon_{t-1}, \mu/b_t  \}$, then, the term $W_t$ is upper bound by a constant and, the radius $\hat{\epsilon}$ asymptotically goes to zero due to the selection as in~\cite[Section IV]{DL-DF-SM:20-lcss}. Consequently, this results in
\begin{equation*}
  \begin{aligned}
    & \liminf\limits_{t \rightarrow \infty} \Prob \left( R_t \leq T F_{t} + a \mu \right) \geq 1-\beta.   \end{aligned} \hfill \textrm{\qed}
\end{equation*}
  \end{proof}

\vspace*{-2ex}

\bibliographystyle{ieeetr}
\bibliography{bib/alias,bib/SMD-add,bib/JC,bib/SM}

\begin{IEEEbiography}[{\includegraphics[width=1in,height=1.25in,clip,keepaspectratio]{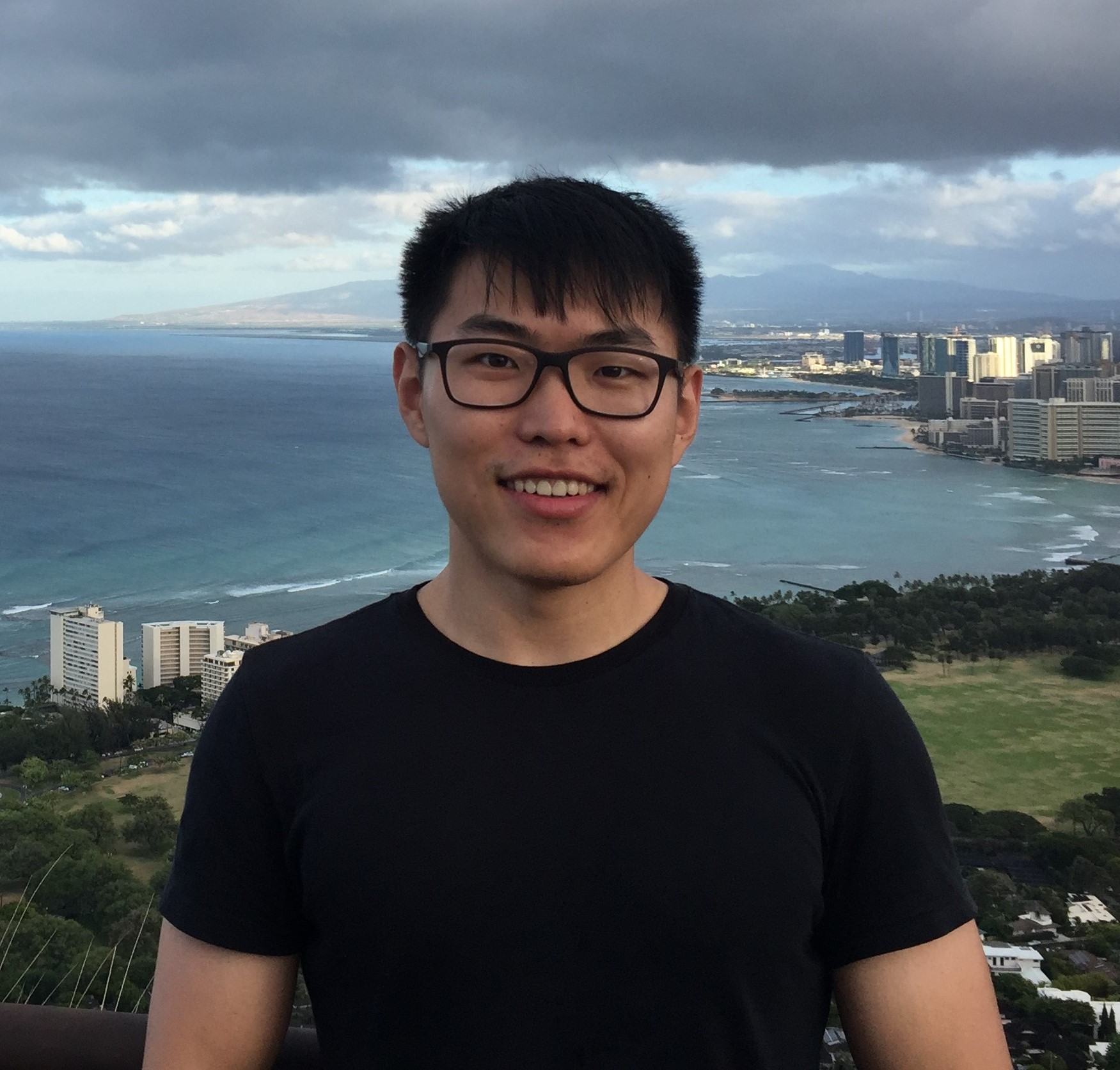}}]{Dan
    Li} receieved the B.E. degree in automation from the Zhejiang
  University, Hangzhou, China, in 2013, the M.Sc. degree in chemical
  engineering from Queen's University, Kingston, Canada, in 2016. He
  is currently a Ph.D. student at University of California, San Diego,
  CA, USA. His current research interests include data-driven systems
  and optimization, dynamical systems and control, optimization
  algorithms, applied computational methods, and stohastic systems. He
  received Outstanding Student Award from Zhejiang University in 2012,
  Graduate Student Award from Queen's University in 2014, and
  Fellowship Award from University of California, San Diego, in 2016.
\end{IEEEbiography}
\begin{IEEEbiography}[{\includegraphics[width=1in,height=1.25in,clip,keepaspectratio]{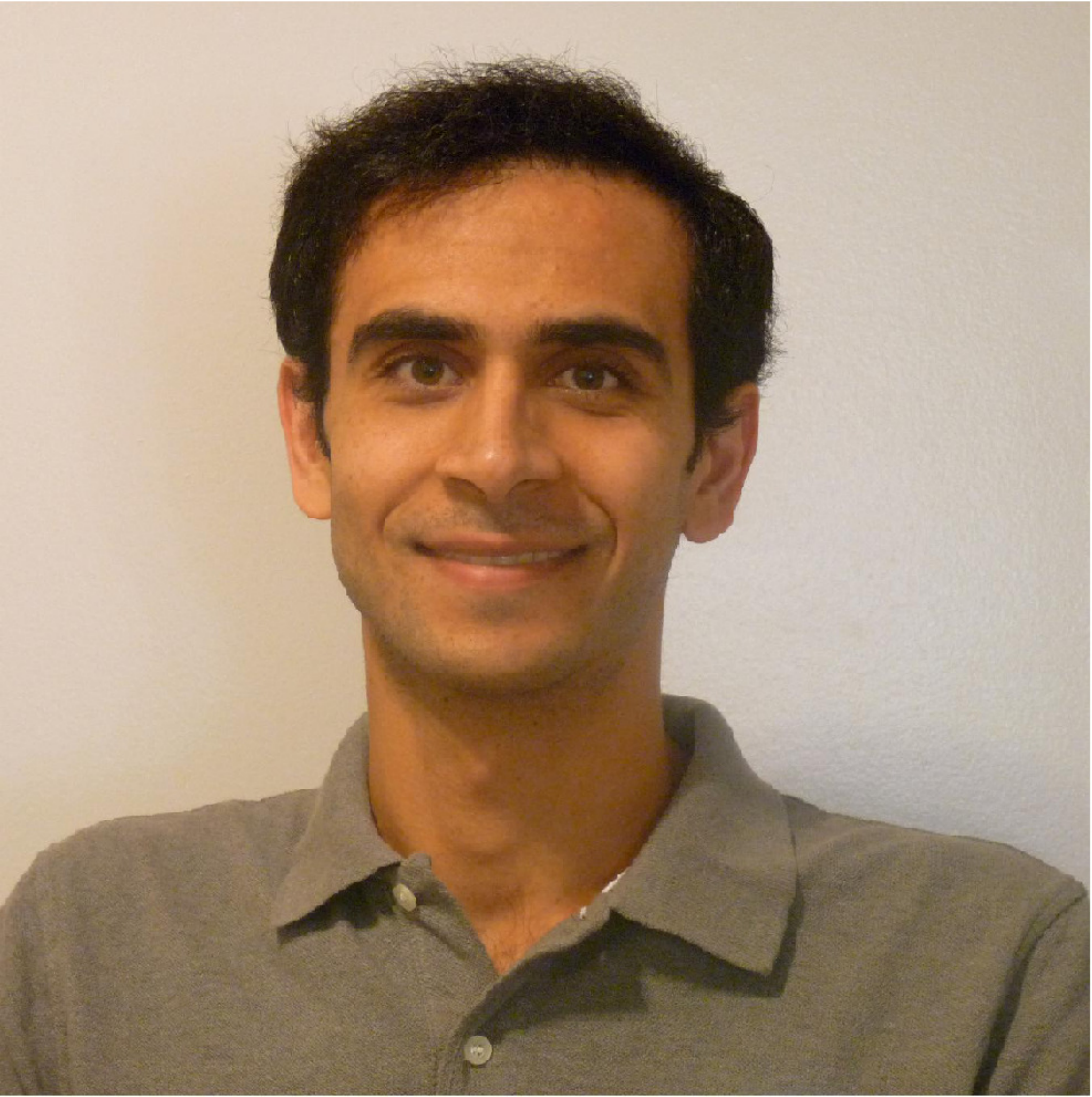}}]{Dariush Fooladivanda}
received the Ph.D. degree from the University of Waterloo, in 2014, and a B.S. degree from the Isfahan University of Technology, all in electrical engineering. He is currently a Postdoctoral Research Associate in the Department of Electrical Engineering and Computer Sciences at the University of California Berkeley. His research interests include theory and applications of control and optimization in large scale dynamical systems.
\end{IEEEbiography}
\begin{IEEEbiography}[{\includegraphics[width=1in,height=1.25in,clip,keepaspectratio]{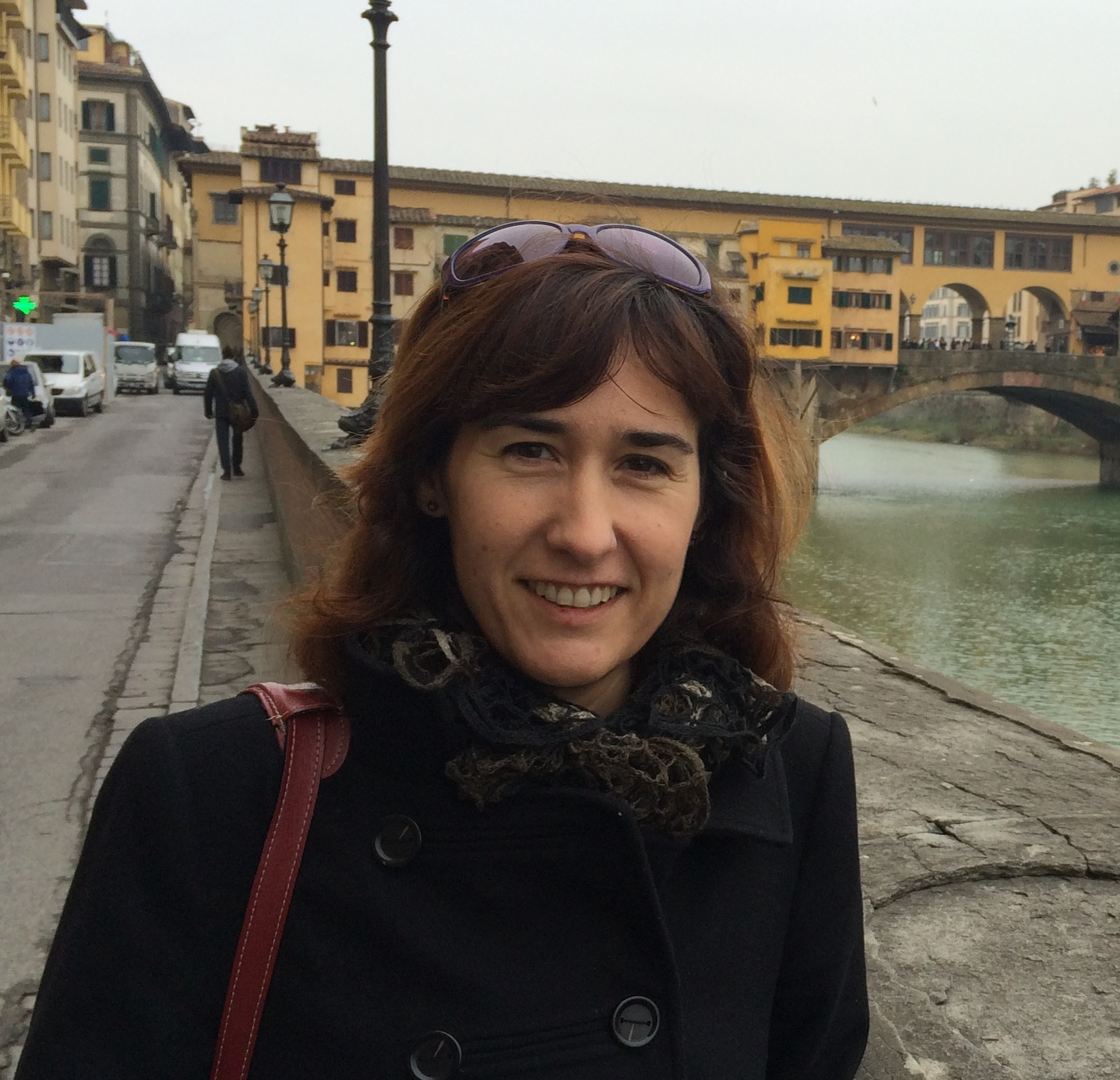}}]{Sonia
   Martínez} is a Professor at the Department of Mechanical and
 Aerospace Engineering at the University of California, San
 Diego. She received her Ph.D. degree in Engineering Mathematics from
 the Universidad Carlos III de Madrid, Spain, in May 2002. Following
 a year as a Visiting Assistant Professor of Applied Mathematics at
 the Technical University of Catalonia, Spain, she obtained a
 Postdoctoral Fulbright Fellowship and held appointments at the
 Coordinated Science Laboratory of the University of Illinois,
 Urbana-Champaign during 2004, and at the Center for Control,
 Dynamical systems and Computation (CCDC) of the University of
 California, Santa Barbara during 2005.

 Her research interests include networked control systems,
 multi-agent systems, and nonlinear control theory with applications
 to robotics and cyber-physical systems. For her work on the control
 of underactuated mechanical systems she received the Best Student
 Paper award at the 2002 IEEE Conference on Decision and Control. She
 co-authored with Jorge Cortés and Francesco Bullo "Motion
 coordination with Distributed Information" for which they received
 the 2008 Control Systems Magazine Outstanding Paper Award. She is a
 Senior Editor of the IEEE Transactions on Control of Networked
 Systems and an IEEE Fellow.
\end{IEEEbiography}

\end{document}